\documentclass[letterpaper]{article} 
\usepackage{aaai23}
\usepackage{times} 
\usepackage{helvet} 
\usepackage{courier} 
\usepackage[hyphens]{url} 
\usepackage{graphicx} 
\usepackage{natbib} 
\usepackage{caption} 
\usepackage{mathtools}
\usepackage{braket}
\usepackage{cancel}

\usepackage{algorithm}
\usepackage{algpseudocode}
\usepackage{xcolor}
\usepackage{dsfont}
\usepackage{hyperref}
\usepackage{soul}

\usepackage{amsmath,amssymb,amsthm}

\newtheorem{theorem}{Theorem}
\newtheorem{lemma}{Lemma}
\newtheorem{corollary}{Corollary}
\usepackage{caption}
\usepackage{subcaption}

\newtheorem{innercustomthm}{Theorem}
\newcommand\numberthis{\addtocounter{equation}{1}\tag{\theequation}}
\newenvironment{customthm}[1]
  {\renewcommand\theinnercustomthm{#1}\innercustomthm}
  {\endinnercustomthm}

\usepackage{newfloat}
\usepackage{listings}
\makeatletter
\newcommand{\setlabel}[1]{\edef\@currentlabel{#1}\label}
\makeatother
\DeclareCaptionStyle{ruled}{labelfont=normalfont,labelsep=colon,strut=off} 
\floatstyle{ruled}
\newfloat{listing}{tb}{lst}{}
\floatname{listing}{Listing}

\nocopyright 

\title{Alternating Layered Variational Quantum Circuits Can Be Classically Optimized Efficiently Using Classical Shadows 

}

\title{Ansatz-Agnostic Exponential Resource Saving in Variational Quantum Algorithms Using Shallow Shadows}
\author {
    Afrad Basheer\footnote{Email: Afrad.M.Basheer@student.uts.edu.au},\textsuperscript{\rm 1}
    Yuan Feng, \textsuperscript{\rm 1}
    Christopher Ferrie, \textsuperscript{\rm 1}
    Sanjiang Li \textsuperscript{\rm 1}
}
\affiliations {
    \textsuperscript{\rm 1} Centre for Quantum Software and Information, University of Technology Sydney, NSW 2007, Australia\\
}

\usepackage{bibentry}
\begin{document}

\maketitle

\begin{abstract}

Variational Quantum Algorithms (VQA) have been identified as a promising candidate for the demonstration of near-term quantum advantage in solving optimization tasks in chemical simulation, quantum information, and machine learning. The standard model of training requires a significant amount of quantum resources, which led us to use classical shadows to devise an alternative that consumes exponentially fewer quantum resources. However, the approach only works when the observables are local and the ansatz is the shallow Alternating Layered Ansatz (ALA), thus severely limiting its potential in solving problems such as quantum state preparation, where the ideal state might not be approximable with an ALA. In this work, we present a protocol based on shallow shadows that achieves similar levels of savings for almost any shallow ansatz studied in the literature, when combined with observables of low Frobenius norm. We show that two important applications in quantum information for which VQAs can be a powerful option, namely variational quantum state preparation and variational quantum circuit synthesis, are compatible with our protocol. We also experimentally demonstrate orders of magnitude improvement in comparison to the standard VQA model.

\end{abstract}

\section{Introduction}
    The fields of quantum computing and quantum algorithms have made huge strides in the past decade. Although we are currently in the era of small erroneous quantum devices called Noisy Intermediate Scale Quantum (NISQ)~\cite{Preskill2018} devices, different research groups were still able to pave the way for demonstrating quantum advantage over classical computers in synthetic but well-defined sampling problems~\cite{Arute2019,Zhong2020,Madsen2022}. The next major breakthrough in this area will be to replicate similar advantages for practically valuable problems. 
    
    Many proposals have been put forward and one class of algorithms that stands out is Variational Quantum Algorithms (VQA)~\cite{McClean2016}. These algorithms are specifically designed to solve optimization problems involving quantum information, which are stored as quantum states using quantum bits a.k.a. qubits and operated using quantum circuits. The core idea is built upon the fact that many important functions involving these objects are notoriously hard or intractable to evaluate on classical computers because this will require classical computational resources exponential in the number of qubits involved. By using parameterized quantum circuits, such functions can be estimated with polynomially many quantum resources on quantum devices, thereby enabling optimization using iterative optimization algorithms. Potentially useful applications include Variational Quantum Eigensolver~\cite{Peruzzo2014}, Quantum Support Vector Machines~\cite{Havlicek2019}, Quantum Approximate Optimization Algorithm~\cite{Farhi2014}, etc.

    Unlike classical computing, the lack of quantum memory devices coupled with the no-cloning theorem implies that each use of a quantum state requires preparing it from scratch. When discussing VQAs, we use the term \textit{sample complexity} to denote the total number of executions of the quantum device required (equivalently the total number of copies of quantum states consumed). In the standard VQA model, this scales linearly with the total number of function evaluations required throughout the optimization. Bring hyperparameter tuning, choice of models and ansatzes, etc. into the picture and suddenly this number is very large.  Moreover, in the near term, only very few capable quantum computers would be available and hence implementing such VQAs with a reduced sample complexity is crucial.

    Interesting parallels can be drawn between VQA training and quantum tomography when viewed in the Heisenberg picture. Since classical shadow tomography~\cite{Huang2020} provides an exponentially better method to estimate linear functionals involving quantum states, this was adopted in the VQA training protocols to achieve an exponential reduction in quantum resources in our previous work~\cite{Basheer2023}. But the method,
    titled \textit{Alternating Layered Shadow Optimization} (ALSO), uses a version of shadow tomography that requires target observables to be local, and this restricts the ansatzes to require simple entanglement structures such as the Alternating Layered Ansatz (ALA) given in Figure~\ref{fig:circuits}(a). This limitation is profound when the optimal circuit or state is not approximable with ALAs. 
    
    The recently proposed shallow shadow technique~\cite{Bertoni2023} describes a similar tomography procedure that can be easily implemented in NISQ devices and does not rely directly on the locality of the observables. By leveraging this, in this work, we introduce \textit{Ansatz Independent Shadow Optimization} (AISO), a method that provides an exponential reduction in quantum resources for VQA training that works with almost all of the popular shallow (depth logarithmic in the number of qubits) quantum circuit structures in the literature, when used in combination with observables of low Frobenius norm. We demonstrate these savings for two important problems in quantum information for which VQAs can be used, namely, Variational Quantum State Preparation (VQSP) and Variational Quantum Circuit Synthesis (VQCS). Both problems concern identifying the right circuit parameters of an ansatz that best approximates unknown quantum states or circuits. 
    
    The benefits of AISO can be summarized as follows:

    \begin{enumerate}

        \item \emph{Exponential saving on input state copies:} To achieve arbitrarily precise estimates of all function evaluations that one encounters during an iterative optimization of the said VQA cost function, AISO consumes exponentially fewer copies of the input state compared to standard VQA, allowing one to do more iterations, achieve better approximations, and carry out extensive hyperparameter tuning.

        \item \emph{Ansatz agnostic implementation on quantum hardware:} Our method guarantees savings of input state copies for almost all the shallow ansatzes used and studied in the literature. Moreover, the operations required using the quantum device are independent of the choice of ansatz.

        \item \emph{Optimization using different ansatzes:} The combination of the two advantages given above means that, for a given unknown input state or circuit, optimization can be carried out with various types of ansatzes. One can then choose the best one that fits, with significant savings in the total usage of quantum devices.

        \item \emph{Compatibility with VQCS:} Solving VQCS involves the usage of maximally entangled states. Since ansatzes with limited entanglement are necessary for ALSO, it cannot be used for efficiently implementing VQCS, This is not the case for AISO, since it is ansatz independent.
    \end{enumerate}

The advantage is experimentally demonstrated in both use cases of interest where we show that AISO outperforms standard VQA significantly given the same number of copies in four different ansatzes used in the literature; Alternating Layered Ansatz (ALA)~\cite{Cerezo2021,Nakaji2021}, Multi-Entanglement Renormalization Ansatz (MERA)~\cite{Bridgeman2015,Franco-Rubio2017,Argüello-Luengo2022}, Hardware Efficient Ansatz (HEA)~\cite{Leone2022,Tang2021,Moreno2023} and Tree Tensor Networks (TTN)~\cite{Shi2006,Nakatani2013,Murg2010} (cf.  Figure~\ref{fig:circuits}).

We also prove that the sample complexity of AISO and, by extension, shallow shadows, can be improved when the input state being sampled is from a $2$-design, instead of a $1$-design. Finally, we argue how AISO is compatible with most of the heuristic methods used to address trainability issues called barren plateaus that one might encounter during optimization.

This paper is organized as follows: in Sections~\nameref{sec:related_works} and~\nameref{sec:background}, we briefly review some related works as well as quantum computing, shallow shadows, and VQA; in Section~\nameref{sec:AISO}, we explain the technical details of AISO; in Section~\nameref{sec:applications}, we discuss VQSP and VQCS; in Section~\nameref{sec:simulation_results}, we present the experimental results comparing AISO with standard VQA; in Section~\nameref{sec:impact_2_design}, we show how one can improve the sample complexity bounds of AISO as well as shallow shadows by assuming that the input is drawn from a state $2$-design rather than a state $1$-design; in Section~\nameref{sec:barren_plateaus}, we explain why AISO is compatible with many of the heuristic barren plateau alleviating techniques in the literature and in Section~\nameref{sec:appendix}, we give the proofs of Theorems~\ref{th:classical_cost}, \ref{th:AISO}, \ref{th:var}, \ref{th:AISO_2design}.

\section{Related Works} \label{sec:related_works}

    Classical shadows have been used to improve the sample complexity of VQAs in our previous work~\cite{Basheer2023}. But in this work, the authors use a type of shadow tomography that relies on the target observables being local. This forces the ansatz to have a weak entanglement structure such as an ALA. So, for applications such as VQSP, results can be poor if the optimal state is not approximable by ALAs. Moreover, this method cannot be used for VQCS since this requires working with the maximally entangled state. Since our method uses shallow shadows, it can be used with almost any shallow ansatz studied in the literature, to solve VQSP and VQCS.

    In \cite{Schreiber2022}, classical shadows have been used to reduce the number of times one has to call a quantum computer in quantum machine learning applications. The idea here is to use the quantum computer to generate classical shadows of an already learned VQA model so that predictions can be made of the learned model using a classical computer. But here, the learning procedure is still carried out on a quantum computer, while in AISO, the learning procedure is carried out completely on a classical computer.

    \begin{figure*}[t] 
    \centering
    \begin{tabular}{ccccc}
         \includegraphics[width=0.2\columnwidth]{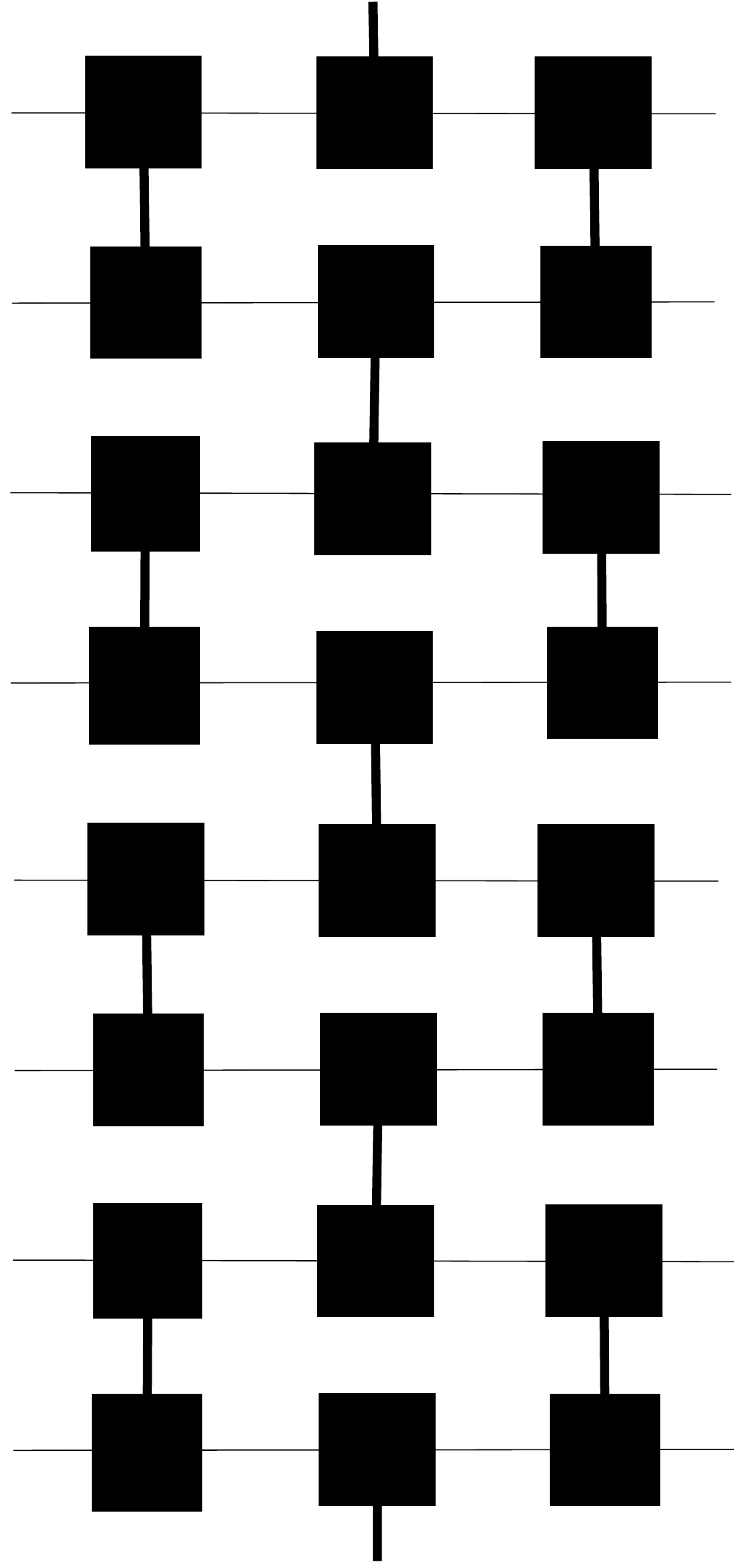} 
         &
        \includegraphics[width=0.36\columnwidth]{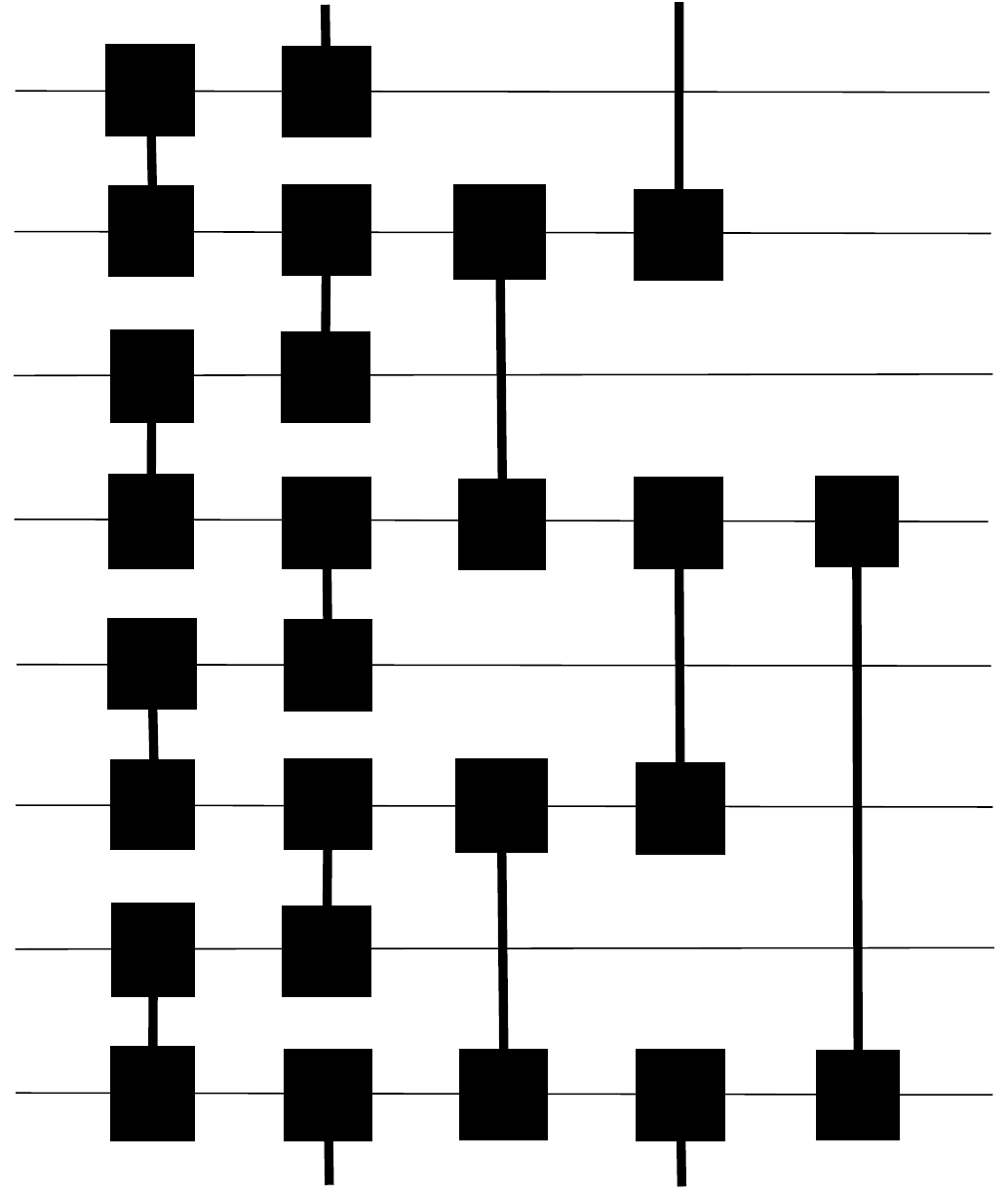}
        &
        \includegraphics[width=0.45\columnwidth]{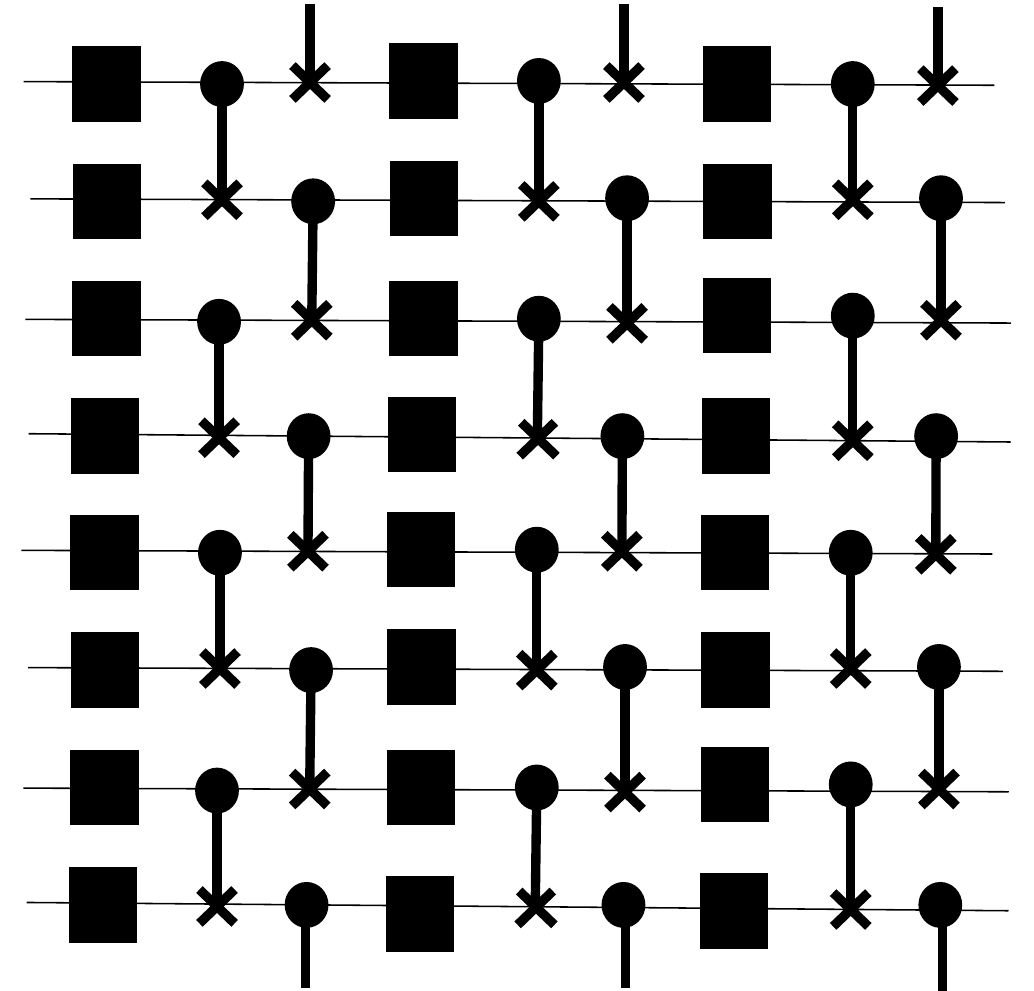}
        &
        \includegraphics[width=0.23\columnwidth]{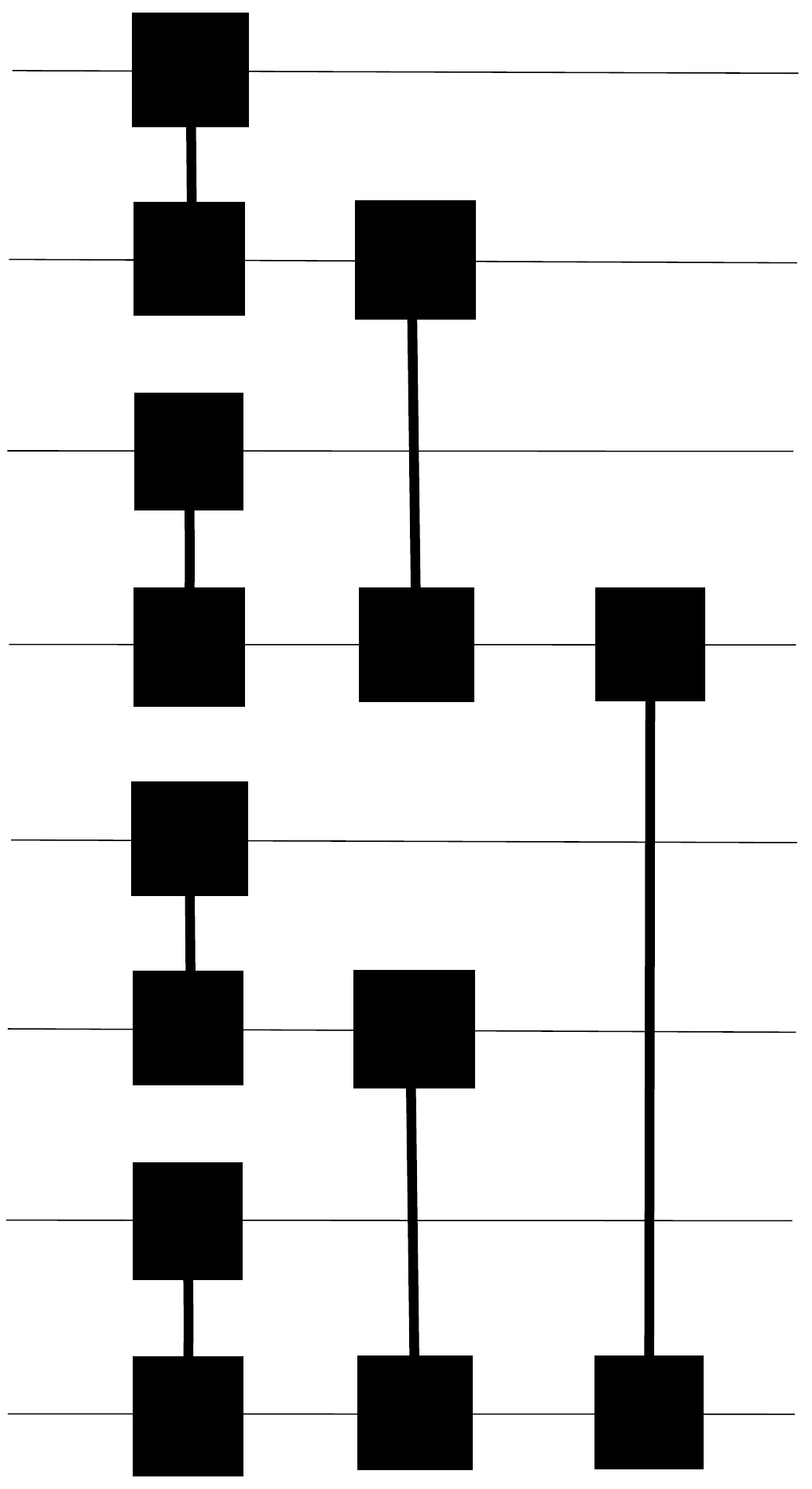} 
        &
        \includegraphics[width=0.3\columnwidth]{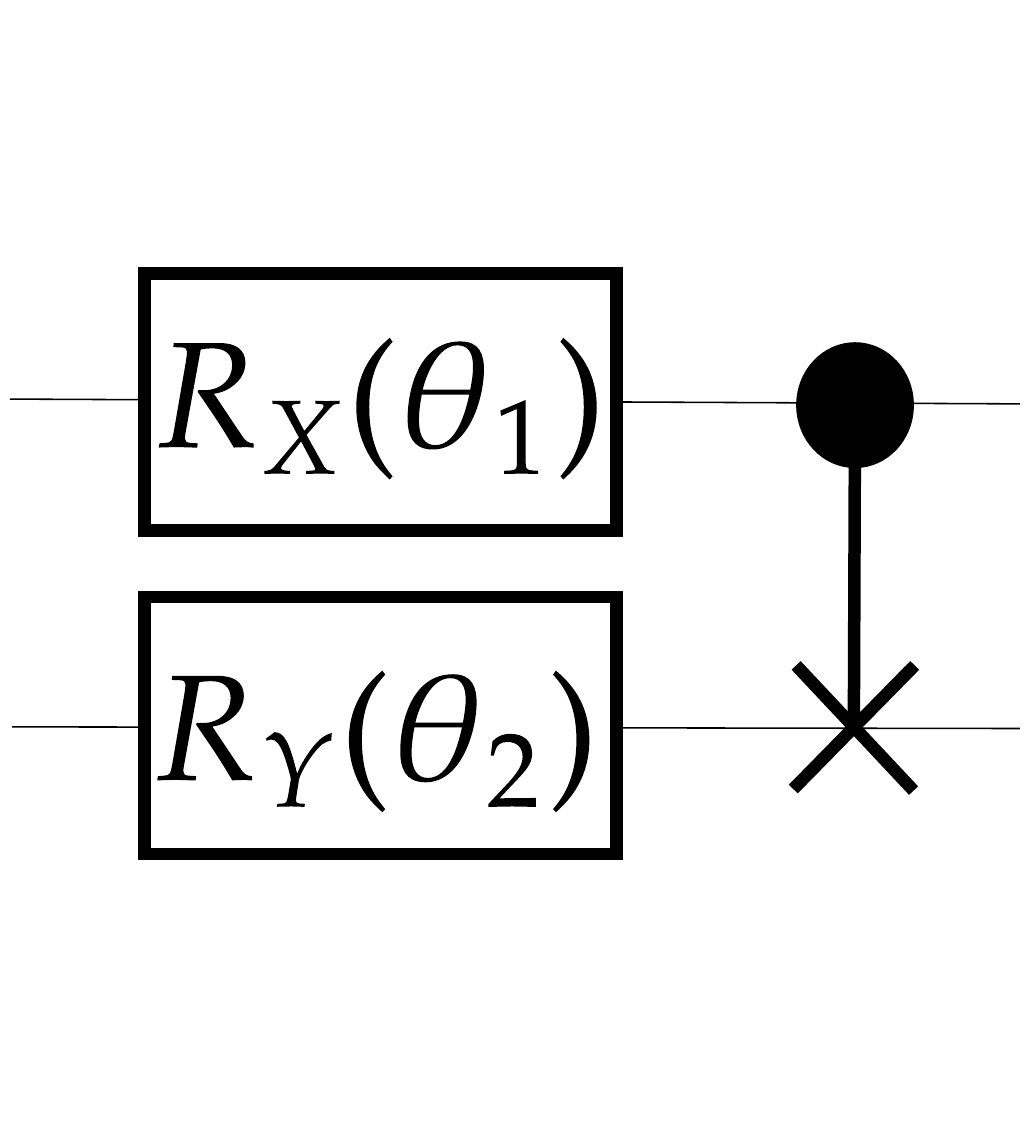}\\\\
        (a) ALA & (b) MERA & (c) HEA & (d) TTN & (e) Subcircuit used in our experiments.
    \end{tabular}
    \caption{Ansatzes used in our simulations. In (a), (b), (d), each connected pair of black boxes represent a two-qubit subcircuit. In (c), each black box is a single qubit subcircuit while the two-qubit gate is the CNOT gate. 
    } \label{fig:circuits}
    \end{figure*}

    In~\cite{Boyd2022}, the ability of classical shadows to estimate an exponentially large number of properties and their classical computational tractability is leveraged to classically approximate VQA cost functions. This is done by simultaneously computing large numbers of covariance functions and using them to solve polynomially growing numbers of root-finding problems. But what is being proposed here is a completely new optimization algorithm, while in AISO, we have a technique that is capable of improving the existing VQA optimization algorithms. So the type of function evaluations that one encounters in AISO will be the same as standard VQA, but we can estimate all of them simultaneously without consuming a lot of copies.

\section{Background} \label{sec:background}
In this section, we review quantum computing, shallow shadows, and VQAs.

%\subsection{Preliminaries}

    \subsection{Quantum Computing}

    Throughout this work, we use the `ket' and `bra' notations to denote column vectors $\ket{\psi}$ and their conjugate transposes $\bra{\psi}$ respectively. $\ket{i} \in \mathbb{C}^d$ is the $i^{\text{th}}$ standard basis vector. We use $\mathcal{L}(\mathbb{C}^d)$ to denote the set of all linear operators that act on $\mathbb{C}^d$.
    
    A quantum \textit{state} is defined as any positive semidefinite operator $\rho \in \mathcal{L}(\mathbb{C}^d)$ with $\text{tr}(\rho) = 1$. In quantum computing, a \textit{qubit} is the analog of a bit in classical computing and can admit any quantum state in $\mathcal{L}(\mathbb{C}^2)$ as its value, The state of an $n$-qubit system can be described using states that act on the tensor product of the $n$ $2$-dimensional vector spaces, denoted as $ \mathbb{C}^2 \otimes \dots \otimes \mathbb{C}^2 \cong \mathbb{C}^{2^n}$. 
    % \yf{check this sentence!} \ab{Changed.} \yf{The `dynamic' should be characterized by some unitary, right?}. \ab{Changed.}

    A unitary operator $U \in \mathcal{L}(\mathbb{C} ^ {2 ^ n})$ is called a \textit{quantum gate} acting on $n$ qubits. Such gates can transform the state of an $n$-qubit system from $\rho$ to $U \rho U^{\dag}$. The \textit{Pauli gates} are defined as 

    \begin{align}
            X = 
            \begin{bmatrix}
                0 & 1 \\
                1 & 0
            \end{bmatrix}, \  
            Y = 
            \begin{bmatrix}
                0 & -i \\
                i & 0
            \end{bmatrix},\ 
            Z = 
            \begin{bmatrix}
                1 & 0 \\
                0 & -1
            \end{bmatrix}.
    \end{align}

    Let $P = \{ \mathds{1}, X, Y, Z\}$ and let $\mathcal{P}^{(\gamma)}_n $ contain all possible distinct $4^n$ $n$-fold tensor products of the elements in $P$, with a scalar $\gamma \in \mathbb{C}$ multiplied to them. Then, the set $\mathcal{P}_n = \mathcal{P}^{(1)}_n \cup \mathcal{P}^{(-1)}_n \cup \mathcal{P}^{(i)}_n \cup \mathcal{P}^{(-i)}_n $ forms a group under matrix multiplication. The normalizer of this group in the group of unitary matrices acting on $\mathbb{C}^{2^n}$ is called the \textit{Clifford group} over $n$ qubits.

    To observe information from a quantum system in a state $\rho$, we measure the system using an \textit{observable}, which is defined as any Hermitian operator $O$. Let the spectral decomposition of $O$ be $O=\sum_i \lambda_i \ket{u_i} \bra{u_i}$. Then the measurement results in an output $\lambda_i$, with probability $ \bra{u_i} \rho \ket{u_i}$. The post-measurement state will be $\ket{u_i}$. In addition, the expected value of this random procedure is $\text{tr}(\rho O)$, concisely written as $\langle O\rangle_{\rho}$. Measurements using diagonal observables are called \textit{standard basis measurements}.

    Rank one states are called \textit{pure states}. In this case, gate operations and measurements can be fully described by any normalized eigenvector in its support.

% \yf{say more about $k$-design: why is it important? Cite some related papers} \ab{Don.e}
% An $n$-qubit operator $V \in \mathcal{L}(\mathbb{C}^{2^n})$ on register $I=[q_1,\ldots,q_n]$ 
% is \textit{$k$-local} if there is a $k$-qubit sub-register $A=[q_{i_1}, q_{i_2},\dots,q_{i_k}]$ such that $V = \widetilde{V} \otimes \mathds1$, where $ \widetilde{V} \in \mathcal{L}(\mathbb{C}^{2^k})$ acts on qubits in $A$ and $\mathds1$ is the identity operator on $I\setminus A$. For simplicity, we often write $V$ as 
% $\widetilde{V}[A]$ and say that $V$ acts \textit{non-trivially} on qubits only in $A$. 

    \begin{figure}[tbh] 
    \centering
    \begin{tabular}{c}
        \includegraphics[width=0.8\columnwidth]{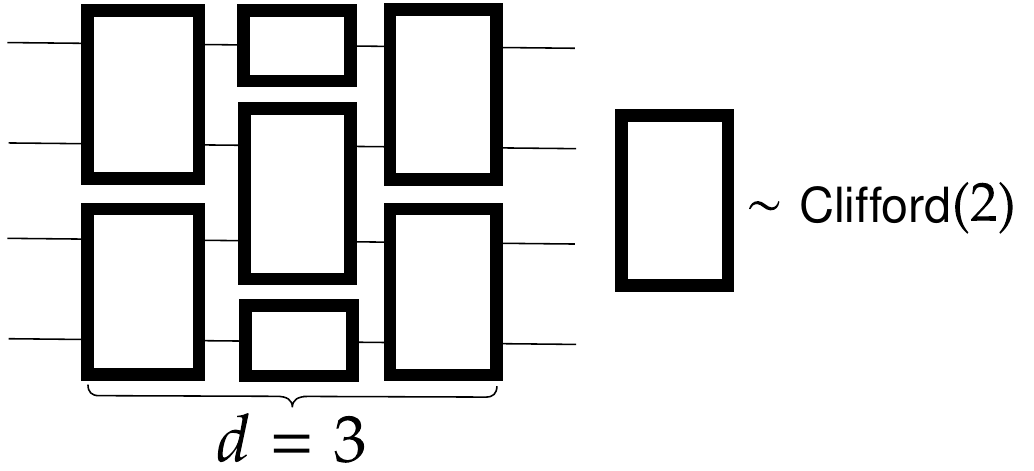}
    \end{tabular}
    \caption{The structure of the unitary ensemble used to generate shallow shadows. Each block here is a uniformly randomly sampled $2$-qubit Clifford circuit. $d$ is the number of vertical layers of these blocks in the circuit.
    } \label{fig:shadow_circuits}
    \end{figure}

    \subsection{Shallow Shadows} \label{subsec:shallow_shadows}

    For an arbitrary state $\rho$ and known observables $O_1, O_2, \dots, O_M$, estimating $\langle O_i \rangle_{\rho}$ for each $i$ using conventional quantum tomography techniques requires $\mathcal{O}(2^n \cdot M)$ copies of $\rho$. Classical shadow tomography~\cite{Huang2020} can be used to estimate all these expectations by consuming only $\mathcal{O}(\log M)$ copies. Moreover, for certain classes of observables, the dependence on $n$ is $\mathcal{O}(\text{poly}(n))$. 
    
    The first step to generate a shadow is to apply a circuit $U$ sampled from an ensemble of $n$-qubit circuits $\mathcal{U}$. Then we measure the resultant state according to the standard basis to obtain an $n$-bit string $u$. Then, a classical shadow is computed classically as 

    \begin{align}
        \hat{\rho}_{U, u} = \Delta_{\mathcal{U}} ^ {-1}(U ^ {\dag} \ket{u} \bra{u} U),
    \end{align}

    where 
    \begin{align}
        \Delta_{\mathcal{U}}(\rho) = \mathbb{E}_{U \sim \mathcal{U}} \sum_{u \in \{ 0,1\}^n} \bra{u} U \rho U^{\dag} \ket{u} U^{\dag} \ket{u} \bra{u} U
    \end{align}
    Furthermore, $\hat{\rho}_{U,u}$ gives an unbiased estimator of $\rho$ and hence $\langle O_i \rangle_{\hat{\rho}_{U,u}}$ is an unbiased estimator of $\langle O_i \rangle_{\rho}$ for all $i$.

    The number of such shadows required for precise estimations is dominated by the \textit{state-dependent shadow norm} of the traceless part of the observables, defined as
    
    \begin{align}
        \| \widetilde{O} \|_{\rho, \mathcal{U}}^2 = \mathbb{E}_{U \sim \mathcal{U}} \sum \limits_{u \in \{ 0,1\}^n} \bra{u} U \rho U ^ {\dag} \ket{u} \langle \widetilde{O} \rangle_{\hat{\rho}_{U,u}} ^ 2,
    \end{align}
    where $\widetilde{O} = O - \frac{\text{tr}(O)}{2 ^ n} \mathds1$. Using this, the sample complexity of the protocol is given by the following theorem.

    \begin{theorem}
        \cite{Huang2020}
        Let $\mathcal{U}$ be an ensemble of gates such that $\Delta_{\mathcal{U}}^{-1}$ exists, and $O_1, O_2, \dots, O_M$ be $n$-qubit observables. For any $\delta, \epsilon \in (0, 1)$, let $T_1=2\log(2M/\delta)$ and $T_2 = (34/\epsilon^2) \max_{i} \| \widetilde{O_i}\|_{\rho,\mathcal{U}}^2$. Let $\rho$ be a state with classical shadows (generated using $\mathcal{U}$) $ \hat{\rho}_{U_1, u_1}, \hat{\rho}_{U_2, u_2}, \dots, \hat{\rho}_{U_{T_1T_2}, u_{T_1T_2}}$.  Define $ \langle \widehat{O_i} \rangle_{\rho} = \mu_{T_1,T_2} (\{ \langle O_i \rangle_{\hat{\rho}_{U_j,u_j}}, \ 1\leq j\leq T_1T_2\})$, where $\mu_{T_1, T_2}$ is the median-of-means estimator (median of $T_1$ means of $T_2$ values each). Then, with probability at least $1-\delta$, we have $ |\langle \widehat{O_i} \rangle_{\rho} - \langle O_i \rangle_{\rho} | \leq \epsilon$ for all $i$.

        \label{th:state_shadow_theorem}
    \end{theorem}

    One way to remove the dependency on $\rho$ and get the worst-case performance guarantees is to replace $\| \widetilde{O}_i\|_{\rho, \mathcal{U}}$ with $\max_{\sigma: \text{state}} \| \widetilde{O_i}\|_{\sigma, \mathcal{U}}$, defined as the \textit{shadow norm}.
    
    In~\cite{Huang2020}, it was shown that when the ensemble is the Clifford group over $n$ qubits, the shadow norm of the observables, and hence the sample complexity, are proportional to the Frobenius norm. But the implementation requires very deep circuits, ruling itself out for NISQ devices.

    Hence in~\cite{Bertoni2023}, the authors propose an ensemble of shallow-depth circuits $\mathcal{U}_d$ (with depth $d$), given in Figure~\ref{fig:shadow_circuits} that achieves similar performance guarantees. Each two-qubit subcircuit here is a uniformly randomly sampled two-qubit Clifford gate. The shadow can be classically computed and stored in the matrix product state form, with cost $\mathcal{O}(2 ^ d)$. Formally, we have

    \begin{theorem} \label{th:shallow_shadows}
        \cite{Bertoni2023}~If $d=\Theta(\log n)$, then for any observable $O$ with $\text{tr}(O)=0$, we have $\| O\|^2_{\mathds1/2^n, \mathcal{U}_d} \leq 4 \| O\|_F^2$.
    \end{theorem}

%    \magenta{(SL: This is confusing as this paper uses the second class.)} \ab{Changed it.}
    
    The term $ \| O\|^2_{\mathds1/2^n, \mathcal{U}_d}$ is called the \textit{locally scrambled shadow norm}. Notice that for any ensemble of states $\mathcal{D}_1$ for which $ \mathbb{E}_{\rho \sim \mathcal{D}_1} (\rho) = \mathds1/2^n$ (also called \textit{state $1$-designs} when all states are pure), $\mathbb{E}_{\rho \sim \mathcal{D}_1} \| O\|_{\rho, \mathcal{U}}^2 = \| O\|_{\mathds1/2^n,\mathcal{U}}^2$ for any gate ensemble $\mathcal{U}$. So, we can view $ \| O\|_{\mathds1/2^n, \mathcal{U}}$ as a quantity that intuitively characterizes the sample complexity of a shadow protocol for a ``typical" state or the performance of the protocol on average, similar to how the shadow norm describes the worst-case performance. This is more apparent when all states in $\mathcal{D}_1$ are pure since then, sampling from $\mathcal{D}_1$ is equivalent to sampling uniformly (according to the spherical measure) from the set of all pure states up to one statistical moment. Moreover, if the observables can be represented using tensor networks with certain properties, then each $\langle O_i \rangle_{\hat{\rho}_{U,u}}$ can be computed classically efficiently.

    \subsection{Variational Quantum Algorithms} \label{subsec:vqa}

Parameterized quantum circuits can be used to encode various optimization problems that one encounters in quantum information. The structure of the circuit used is called an \textit{ansatz}. We use $U(\boldsymbol{\theta})$ to denote a parameterized circuit, where $\boldsymbol{\theta}$ is a vector of parameters. In standard VQA, we use $U(\boldsymbol{\theta})$ to estimate the value of a target function and then optimize the parameters by feeding the output to a classical iterative optimizer.

For any ansatz $U$, we define $\rho(\boldsymbol{\theta}) \coloneqq U(\boldsymbol{\theta}) \rho U(\boldsymbol{\theta}) ^ {\dag}$.  Our focus in this paper is on the function defined (over $\boldsymbol{\theta}$) as
\begin{align}  \label{eq:general_optimization}
    \langle O \rangle_{\rho(\boldsymbol{\theta})} = \text{tr}(U(\boldsymbol{\theta}) \rho U(\boldsymbol{\theta})^{\dag} O),
\end{align}
where $\rho$ is the input quantum state and $O$ is an output observable, and we aim to find the parameters that maximize it. One can estimate $\langle O \rangle_{\rho(\boldsymbol{\theta})}$ for any $\boldsymbol{\theta}$ by repeated measurements, after the application of $U(\boldsymbol{\theta})$ on $\rho$. Given this ability, the gradient of $ \langle O \rangle_{\rho(\boldsymbol{\theta})}$ can also be estimated using standard methods such as finite differencing or quantum-specific approaches such as the parameter shift rule~\cite{Mitarai2018}. Problems in quantum information that can be reduced to an instance of optimization of Eq~\eqref{eq:general_optimization} include variational quantum eigensolver~\cite{Peruzzo2014}, quantum autoencoder~\cite{Romero2017}, as well as VQSP and VQCS.

\section{Ansatz Independent Shadow Optimization} \label{sec:AISO}

    In this section, we explain the main idea and theoretical results behind AISO.

    \subsection{Method} \label{subsec:methods}

        For any quantum circuit $V$, for any qubit $i$, we define the number of times a gate touches or crosses the qubit wire as $R_{V,i}$. Let $R_V = \max_i R_{V,i}$. We require our ansatz $U$ to have $ R_U \in \mathcal{O}(\log n)$. Note that most shallow ansatzes used in the literature will satisfy this. Let $\langle O\rangle_{\rho(\boldsymbol{\theta} ^ {(1)})},\langle O\rangle_{\rho(\boldsymbol{\theta} ^ {(2})}, \dots, \langle O\rangle_{\rho(\boldsymbol{\theta} ^ {(C)})}$ be function evaluations that one encountered while optimizing Eq~(5) using an iterative optimization algorithm.
        
        Define $W_{O}(\boldsymbol{\theta}) = U(\boldsymbol{\theta})^{\dag} O U(\boldsymbol{\theta})$. Each function evaluation can be seen as estimating the expectation of $\rho$ with these parameterized observables because 
        \begin{align}\label{eq:costfun}
            \langle O \rangle_{\rho(\boldsymbol{\theta})} =  \text{tr}( U(\boldsymbol{\theta}) \rho U(\boldsymbol{\theta}) ^ {\dag} O) = \langle W_{O}(\boldsymbol{\theta}) \rangle_{\rho}.
        \end{align}

        Moreover, the Frobenius norm remains invariant since $\| O\|_F^2 = \| V O V^{\dag}\|_F^2$ for any unitary $V$.

        Now, using Theorems~\ref{th:state_shadow_theorem} and~\ref{th:shallow_shadows}, we can estimate all $C$ function evaluations using shallow shadows, and the AISO protocol goes as follows.

        \begin{enumerate}
            \item Load $T_1T_2 $ shallow shadows of $\rho$, where $T_1 = \mathcal{O}(\log C)$ and $ T_2= \mathcal{O}(\| O\|_F^2)$. Let them be $\hat{\rho}_{U_1, u_1}, \hat{\rho}_{U_2, u_2}, \dots, \hat{\rho}_{U_{T_1T_2}, u_{T_1T_2}}$

            \item Use the iterative optimization algorithm to optimize the target function             
            \begin{equation}\label{eq:tfun}
            \langle \widehat{W_O}(\boldsymbol{\theta})\rangle_{\rho} \coloneqq \mu_{T_1, T_2}(\{ \langle W_O(\boldsymbol{\theta})\rangle_{\hat{\rho}_{U_j,u_j}} \ 1\leq j\leq T_1T_2\}).
            \end{equation}
            
        \end{enumerate}

        The cost of classical computation is dominated by the cost of computing $ \langle \widehat{W_O}(\boldsymbol{\theta})\rangle_{\rho}$ classically. In this case, we have

        \begin{theorem} \label{th:classical_cost}
            In AISO, for any quantum ansatz $U$ with $R_U \in \mathcal{O}(\log n)$,  $\langle W_O(\boldsymbol{\theta}) \rangle_{\rho}$ can be classically evaluated with cost $\mathcal{O}(\text{poly}(n) \cdot \log C \cdot \| O\|_F^2) $ for VQSP and VQCS. 
        \end{theorem}

    \subsection{Sample Complexity}
        In this section, we prove the bounds on the sample complexity when the input state is sampled from a state $1$-design. The goal is to show that on average, AISO requires only a number of copies that is logarithmically dependent on $M$ and linearly dependent on $\| O\|_F^2$.

        \begin{theorem} \label{th:AISO}
            Let $d = \Theta(\log n)$ and $\rho$ be an $n$-qubit pure state sampled from a state 1-design $\mathcal{D}_1$. For any $\delta, \epsilon \in (0,1)$, $m > 1/\delta$, and any $C>0$, let
            \begin{align}\label{eq:T_11design}
            T_1 \geq 2 \log \left(\frac{2(m-1)C}{m\delta-1}\right),\ T_2 \geq \frac{136}{\epsilon^2} m \| O\|_F^2.
            \end{align} 
            
            Then for any parameter vectors $\boldsymbol{\theta} ^ {(1)}, \boldsymbol{\theta} ^ {(2)}, \dots, \boldsymbol{\theta} ^ {(C)}$, all values $\langle W_O(\boldsymbol{\theta} ^ {(c)}) \rangle_{\rho}$, $1\leq c\leq C$, defined as in Eq~\eqref{eq:costfun} can be estimated using $ \langle \widehat{W_O}(\boldsymbol{\theta} ^ {(c)}) \rangle_{\rho}$ defined as in Eq~\eqref{eq:tfun} so that with probability at least $1-\delta$, we have $| \langle W_O(\boldsymbol{\theta} ^ {(c)})\rangle_{\rho} - \langle \widehat{W_O}(\boldsymbol{\theta} ^ {(c)}) \rangle_{\rho} | \leq \epsilon$
            for all $c$.
        \end{theorem}

        Since estimating $C$ evaluations in standard VQA requires preparing $U(\boldsymbol{\theta} ^ {(c)})$ for all $c$ and measuring each of them multiple times, the total number of copies required would be $\mathcal{O}(C)$, which is exponentially higher than AISO. One key reason why this is the case is that in standard VQA, we cannot reuse the measurement results, since each of them was conducted specifically to estimate $\langle W_O(\boldsymbol{\theta} ^ {(c)})\rangle_{\rho} $ for some $c$. Meanwhile, in AISO, all quantum measurements made are \textit{independent} of $\boldsymbol{\theta}^{(c)}$, and all these measurements are used while estimating all the expectations.

   %thus resulting in a classical memory for efficiently storing the quantum states, for use in alternating layered VQAs.

    Although the constants appear large, since we use union bounds as well as a few loose constants in Theorem~\ref{th:shallow_shadows}, in practice significantly lesser copies than what is suggested in Theorem~\ref{th:AISO} suffice. We explore this in detail in our experimental results. The space complexity of the protocol is dominated by the storage of shallow shadows. Each shadow is an MPS with maximum bond dimension at most $2 ^ {d-1}$. This means that each shadow can be stored using at most $n2^{d}$ complex numbers and hence the total space complexity is $nT_1T_22^{d}$. So, when $d = \mathcal{O}(\log n)$, the space complexity is $\mathcal{O}(\text{poly}(n) \cdot T_1T_2)$.

    \begin{figure*}[tbh]
    \centering
        \begin{tabular}{cccc}
        ALA & MERA & HEA & TTN \\
         \includegraphics[width=0.5\columnwidth]{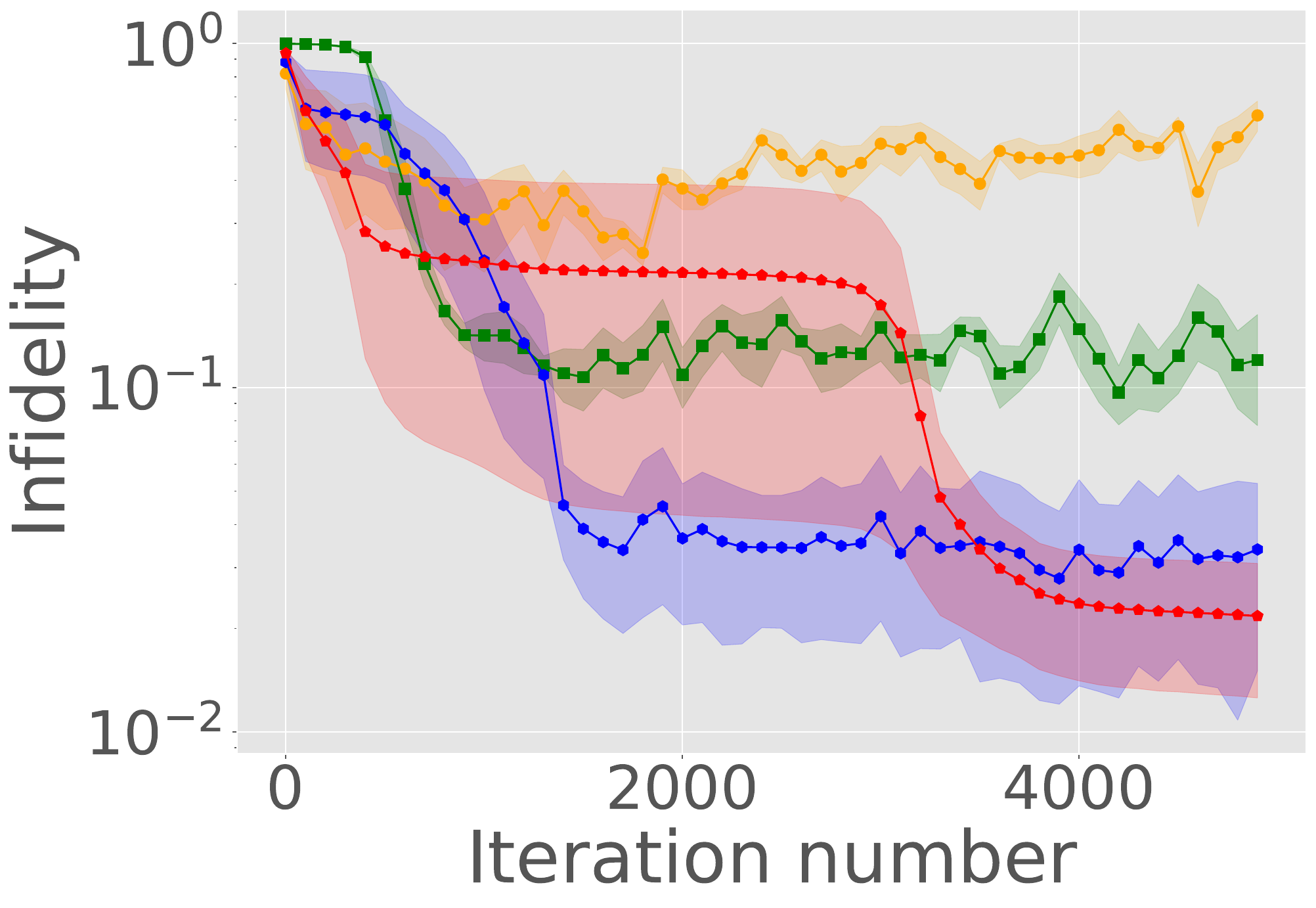} 
         &
        \includegraphics[width=0.5\columnwidth]{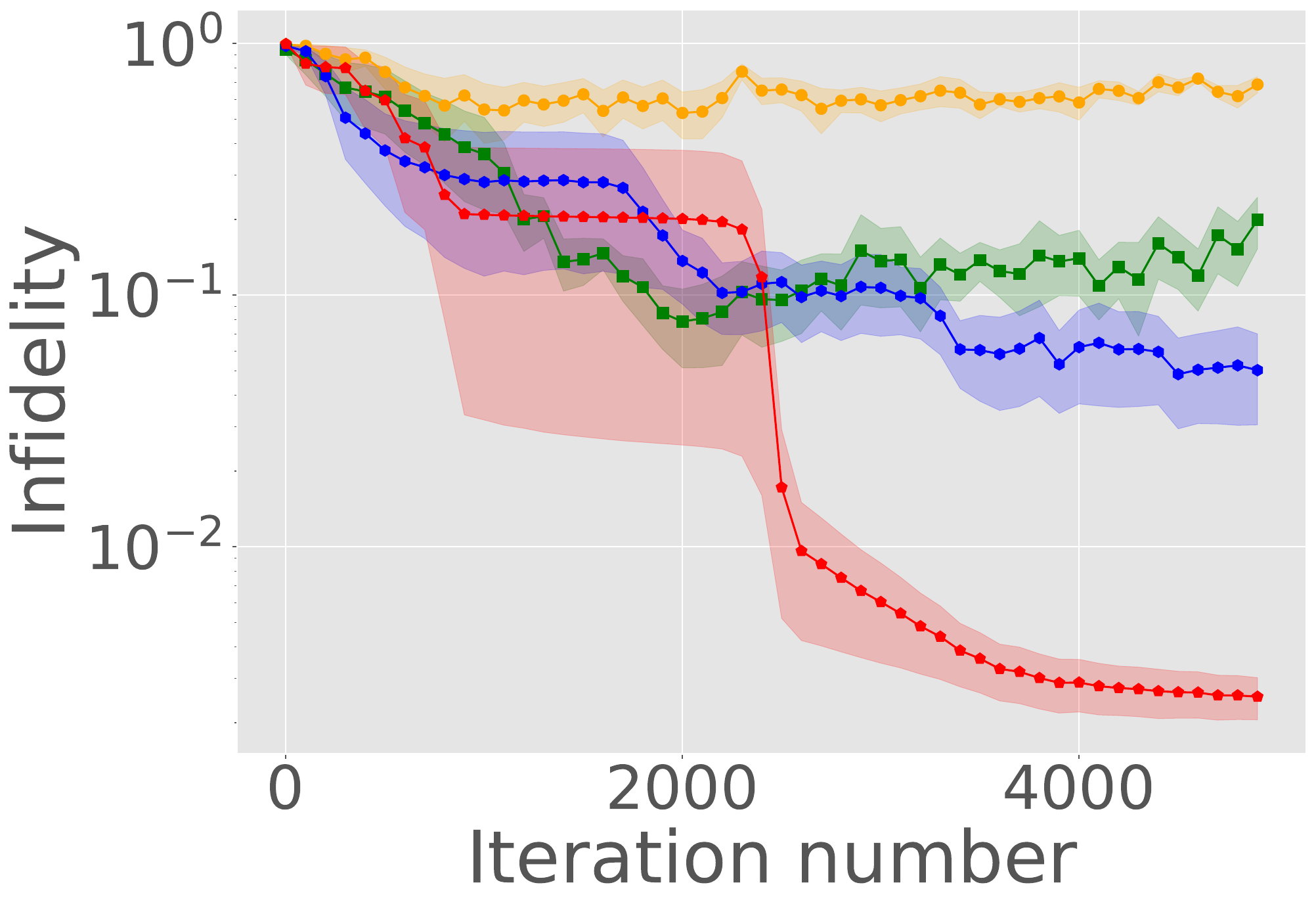}
        &
        \includegraphics[width=0.5\columnwidth]{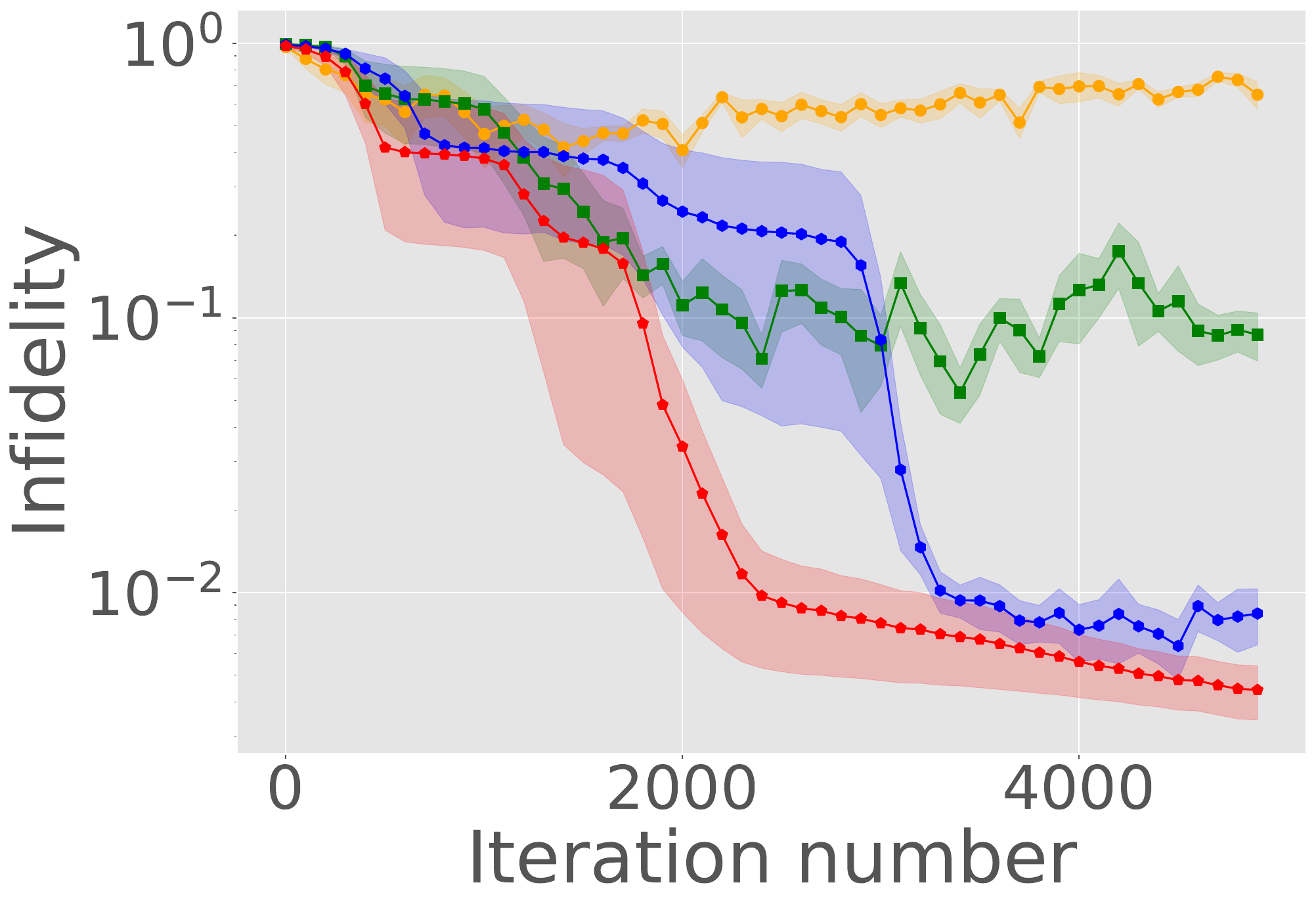}
        &
        \includegraphics[width=0.5\columnwidth]{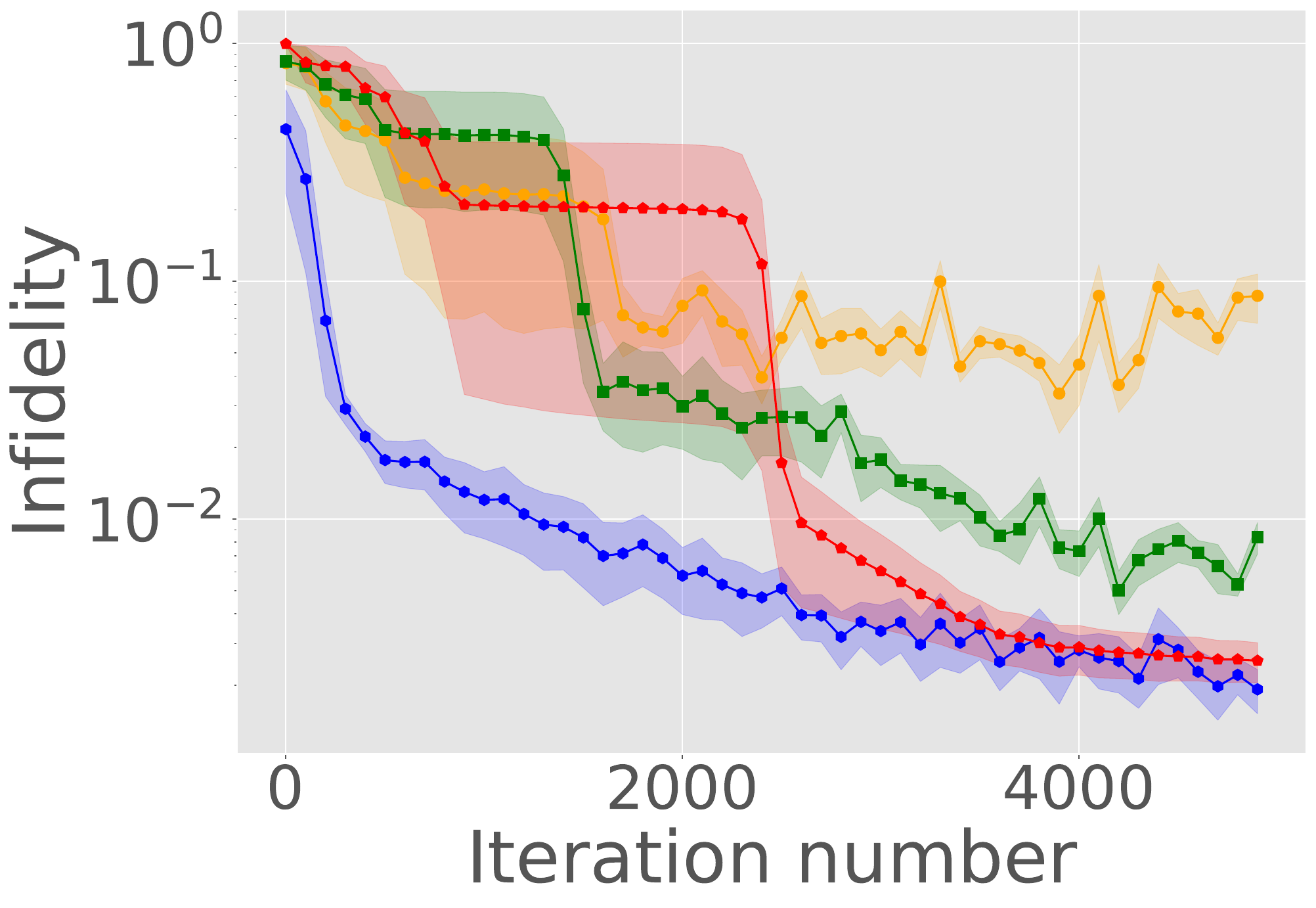} \\
        (a) & (b) & (c) & (d) \\
         \includegraphics[width=0.5\columnwidth]{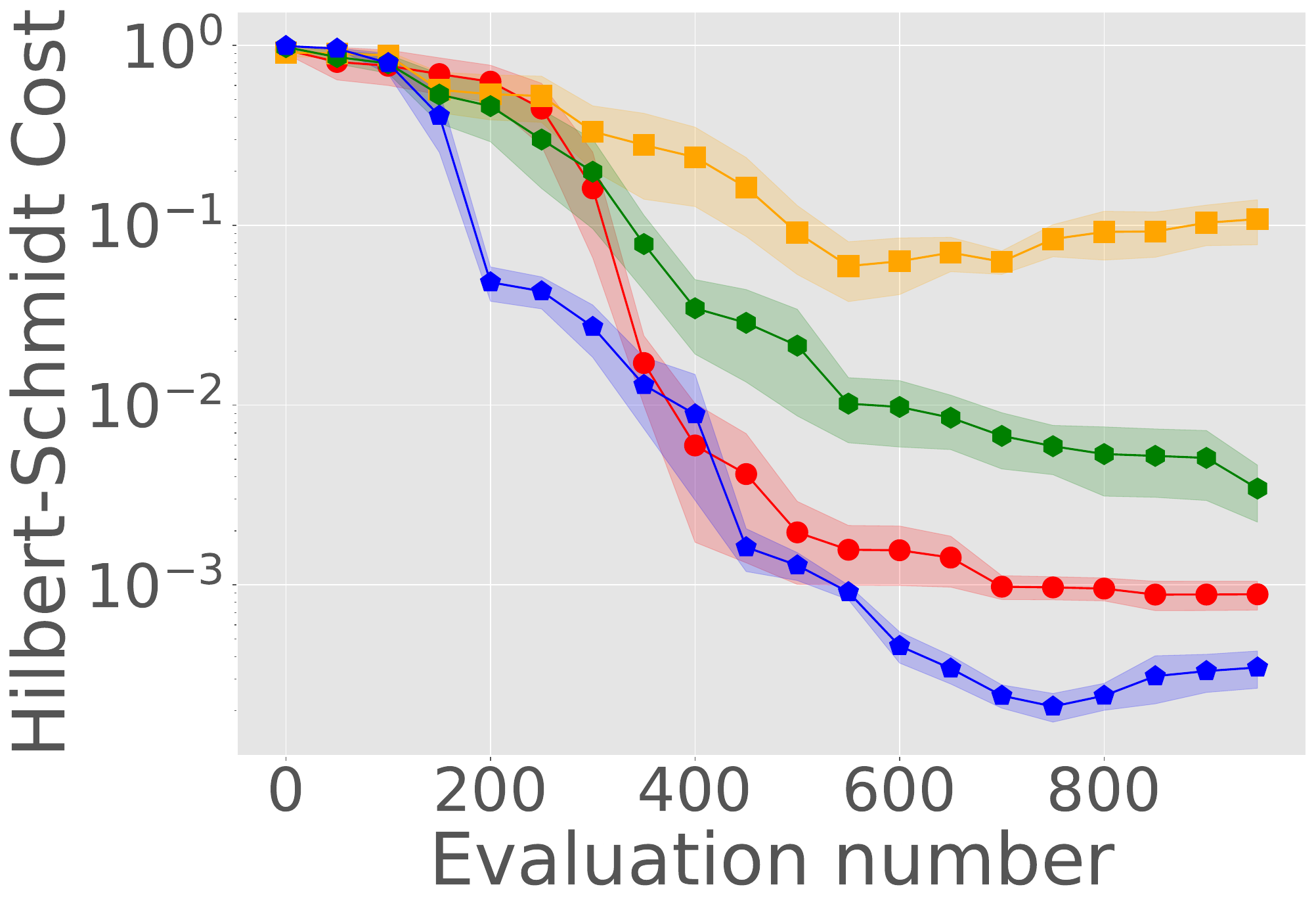} 
         &
        \includegraphics[width=0.5\columnwidth]{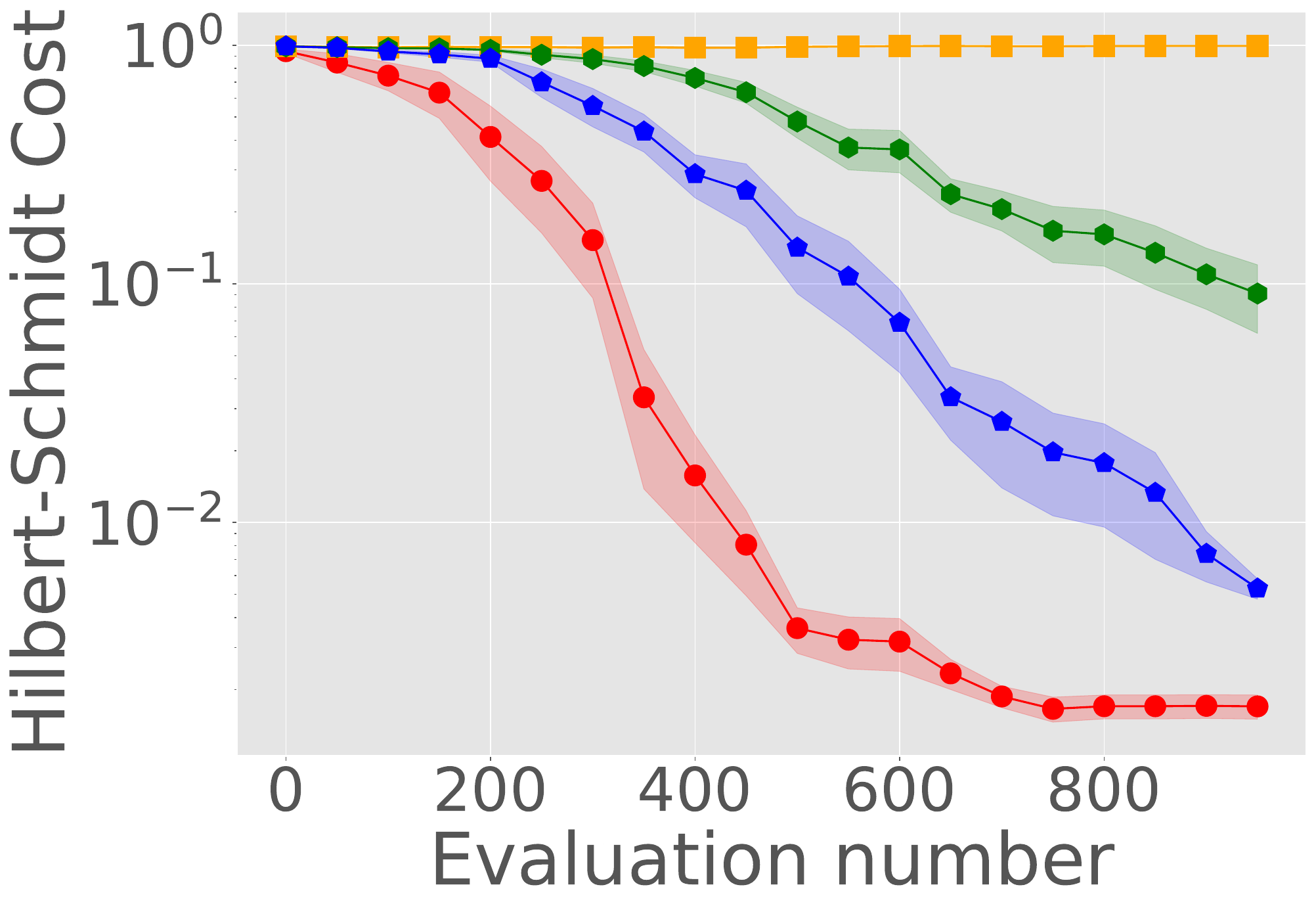}
        &
        \includegraphics[width=0.5\columnwidth]{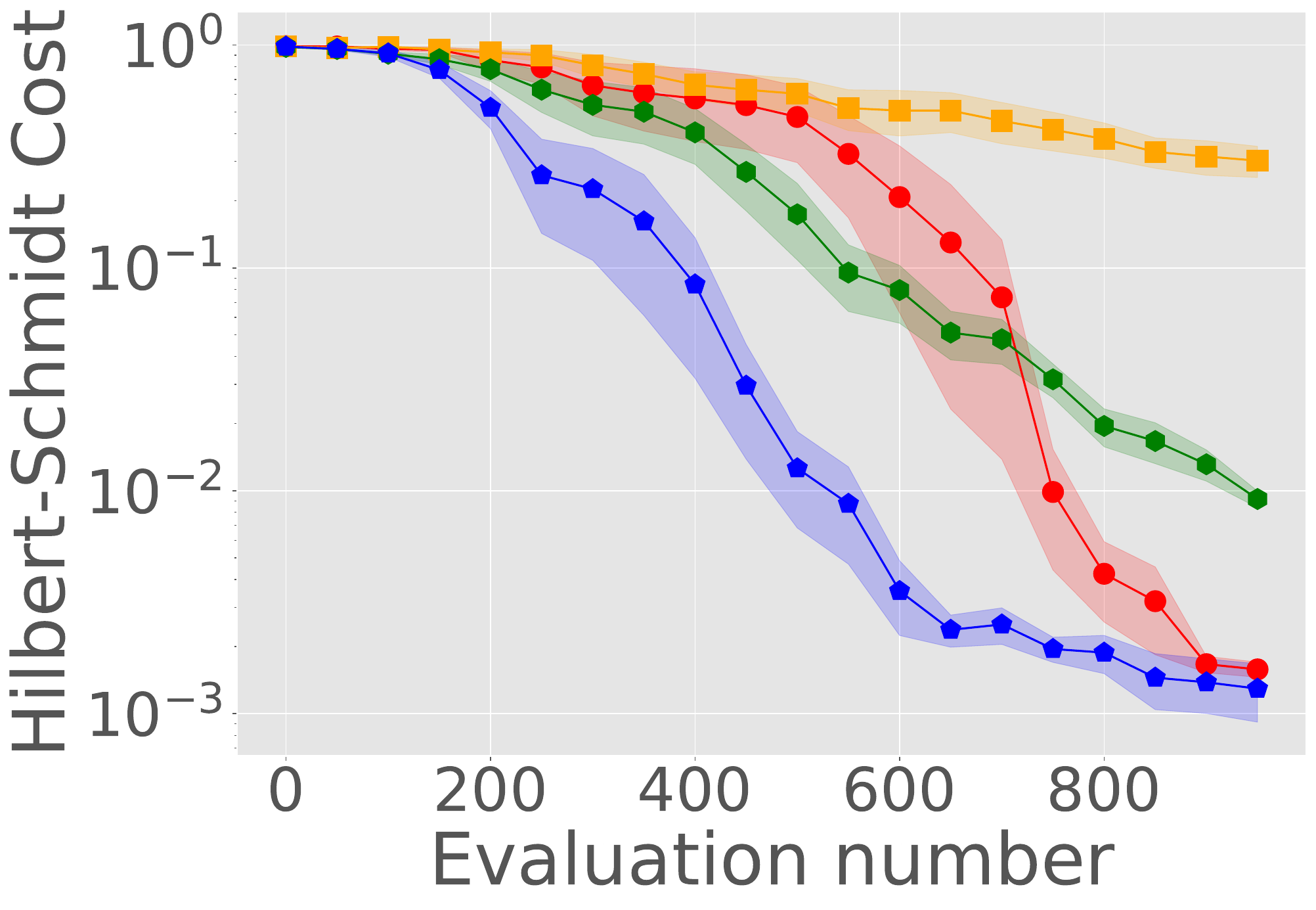}
        &
        \includegraphics[width=0.5\columnwidth]{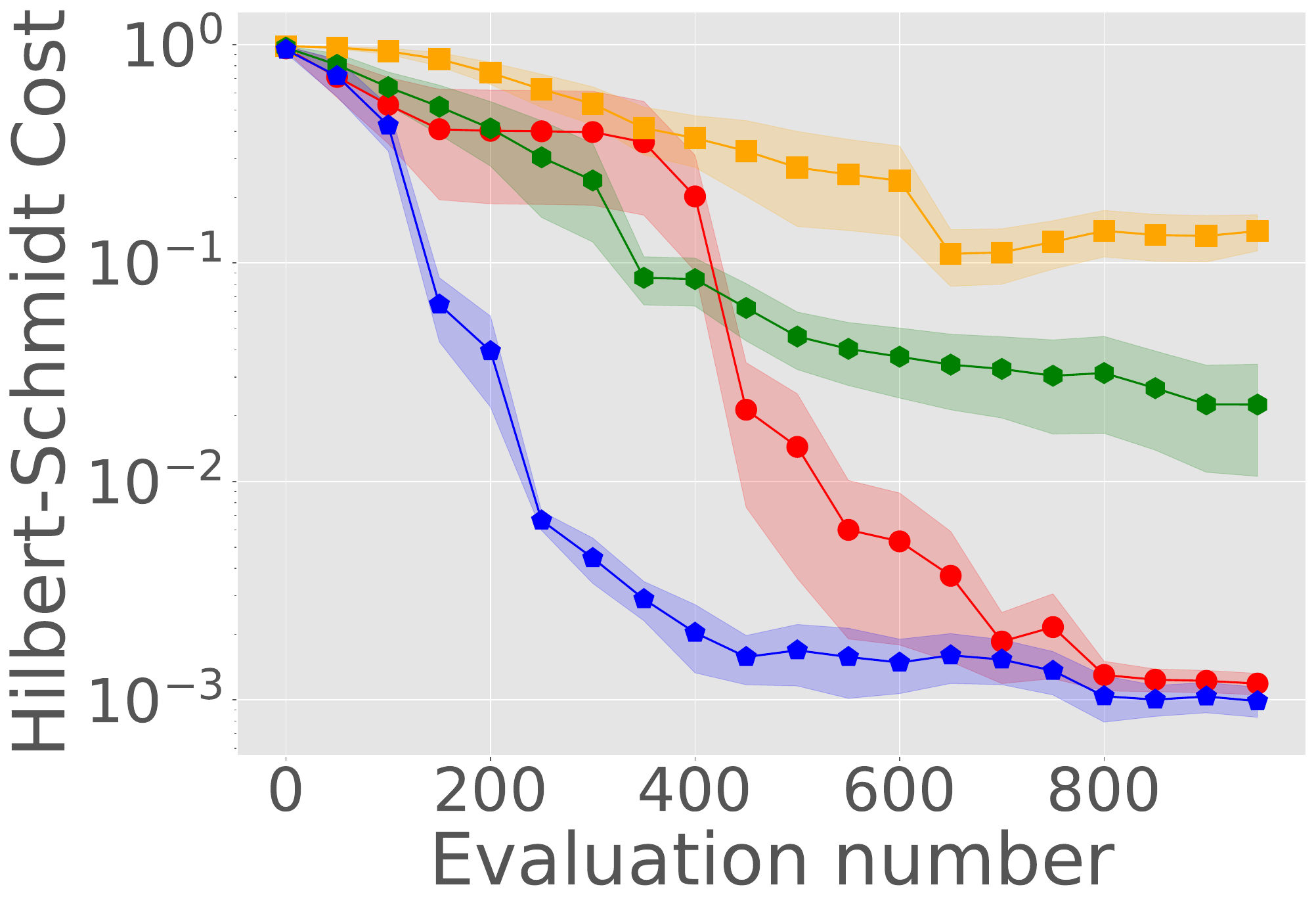} \\
        (e) & (f) & (g) & (h)
        \end{tabular}
        \caption{Simulation results comparing the learning curves of AISO with the standard VQA. Each shaded region corresponds to $5$ instances of a problem. VQA/AISO $(T)$ consumes $T$ copies in total throughout the optimization. Plots (a,e), (b,f), (c,g), (d,h) correspond to ALA, MERA, HEA, and TTN being used as the ansatz, respectively. In plots (a-d), we compare the learning rates of AISO with standard VQA in VQSP. The classical optimizer used is the SPSA algorithm, with $5\times 10^3$ iterations. The red curve represents AISO $(10^4)$ while the orange, green, and blue curves represent VQA $(5 \times 10^5)$, VQA $(10^6)$, and VQA $(2.5 \times 10^6)$ consuming $50, 100, $and $250$ copies per function evaluation respectively. Similarly, in plots (e-h), we compare the learning rates of AISO with standard VQA in VQCS. The classical optimizer used is Powell's method, with a total of $10^3$ function evaluations allowed. In these plots, the minimum Hilbert-Schmidt Cost in each interval of $100$ function evaluations is plotted. Like previously, the red curve represents AISO $(10^4)$ while the orange, green, and blue curves represent VQA $(10^5)$, VQA $(10^6)$, and VQA $(10^7)$ consuming $10, 10^2$, and $10^3$ copies per function evaluation respectively. We can see that AISO can closely match or outperform standard VQA by consuming orders of magnitude fewer copies in total.} \label{fig:fin_vs_AISO}
    \end{figure*}
    
\section{Applications} \label{sec:applications}

    In this section, we discuss how AISO can be used to tackle two important problems in quantum information.

    \subsection{Variational Quantum State Preparation }

    In VQSP, our goal is to find a circuit that is capable of preparing (approximately) a pure state $\rho = \ket{\psi} \bra{\psi}$, given access to multiple copies of it. That is, we would like to find a parameter vector $\boldsymbol{\theta}$ that minimizes the \textit{infidelity} between $U(\boldsymbol{\theta}) ^ {\dag} \ket{0}$ and $\ket{\psi}$, defined as $ 1 - |\bra{\psi} U(\boldsymbol{\theta})^{\dag} \ket{0}|^2$, where $U$ is a heuristically chosen ansatz. Infidelity assumes values in $[0,1]$ and is widely used in quantum information to measure how far apart two states are, $1$ implying orthogonality and $0$ implying equality. Note that the minimization of infidelity is the same as the maximization of $ \langle \ket{0} \bra{0}\rangle_{\rho(\boldsymbol{\theta})}$. Since $\ket{0} \bra{0}$ has unit Frobenius norm, this objective function is compatible with AISO. Also, using AISO, one can attempt to find the best parameters for a wide variety of circuit ansatzes using multiple optimization procedures with very few copies consumed.
    
    Moreover, for any shallow shadow $ \hat{\rho}$, $\langle W_O(\boldsymbol{\theta}) \rangle_{\hat{\rho}}$ can be computed classically efficiently by contracting the tensor network given in Figure~\ref{fig:tensor_networks}(a).  Even though the example given here is the ALA, using Theorem~\ref{th:classical_cost}, one can easily replace it with any ansatz with $R_U \in \mathcal{O}(\log n)$.

    \subsection{Variational Quantum Circuit Synthesis } 
        VQCS is a natural extension of VQSP to quantum circuits. Here, our goal is to learn the parameters of an $n$-qubit ansatz $U(\boldsymbol{\theta})$ that best approximates a given unknown quantum gate $V$. Similar to how we use infidelity for quantum states, we can use the Hilbert-Schmidt cost function defined for unitaries in~\cite{Khatri2019}. For any $\boldsymbol{\theta}$, this is computed as $H(\boldsymbol{\theta})=1 - 1/4^n|\text{tr}(U(\boldsymbol{\theta}) ^ {\dag} V)|^2$ and minimizing $H$ gives us the set of parameters that prepares (approximates) $V$. 
        
        To see why, first note that any quantum gate $W$ can be uniquely identified using a representation given as $W \otimes \overline{W}$. This can be derived from its action on the vectorized version of elements in $\mathcal{L}(\mathbb{C}^{2^n})$. Then we see that $H(\boldsymbol{\theta})$ is proportional to $ \| U(\boldsymbol{\theta}) \otimes \overline{U(\boldsymbol{\theta})} - V \otimes \overline{V}\|_F^2$.

        To evaluate $H(\boldsymbol{\theta})$ for any $ \boldsymbol{\theta}$, we start with the maximally entangled state on two $n$-qubit systems, defined as $\ket{\Phi} = 1/\sqrt{2^n}\sum_{i=0}^{2^n-1} \ket{i} \ket{i}$. Then, we apply $V$ on the second register to obtain $\ket{V} = 1/\sqrt{2^n}\sum_{i=0}^{2^n-1} \ket{i} \ket{v_{\bullet i}}$, where $\ket{v_{\bullet i}}$ is the $i^{\text{th}}$ column of $V$. Then, one can see that $H(\boldsymbol{\theta})=1 -  \langle \ket{U(\boldsymbol{\theta})} \bra{U(\boldsymbol{\theta})} \rangle_{\ket{V} \bra{V}}$. Therefore, we can use shallow shadows of $\ket{V}$ to estimate $H(\boldsymbol{\theta})$. Since $ \| \ket{U(\boldsymbol{\theta})} \bra{U(\boldsymbol{\theta})} \|_F = 1$ for all $\boldsymbol{\theta}$, the number of shadows, or equivalently, the number of applications of $V$, is independent of $n$. 
        
        In terms of classical computational complexity, $ \langle \ket{U(\boldsymbol{\theta})} \bra{U(\boldsymbol{\theta})} \rangle_{\hat{\rho}}$ for any shallow shadow $\hat{\rho}$ can be computed by contracting the tensor network given in Figure~\ref{fig:tensor_networks}(b), the cost of which is polynomial in $n$. The explanation regarding the usage of ALA in this figure is the same as the one for VQSP. From now on, when discussing the sample complexity of VQCS, the ``number of copies" will mean the number of copies of $\ket{V}$ consumed (equivalently, the number of applications of $V$.)
    \label{subsec:csp}

    \begin{figure*}[tbh]
    \centering
        \begin{tabular}{cccc}
        ALA & MERA & HEA & TTN \\
         \includegraphics[width=0.5\columnwidth]{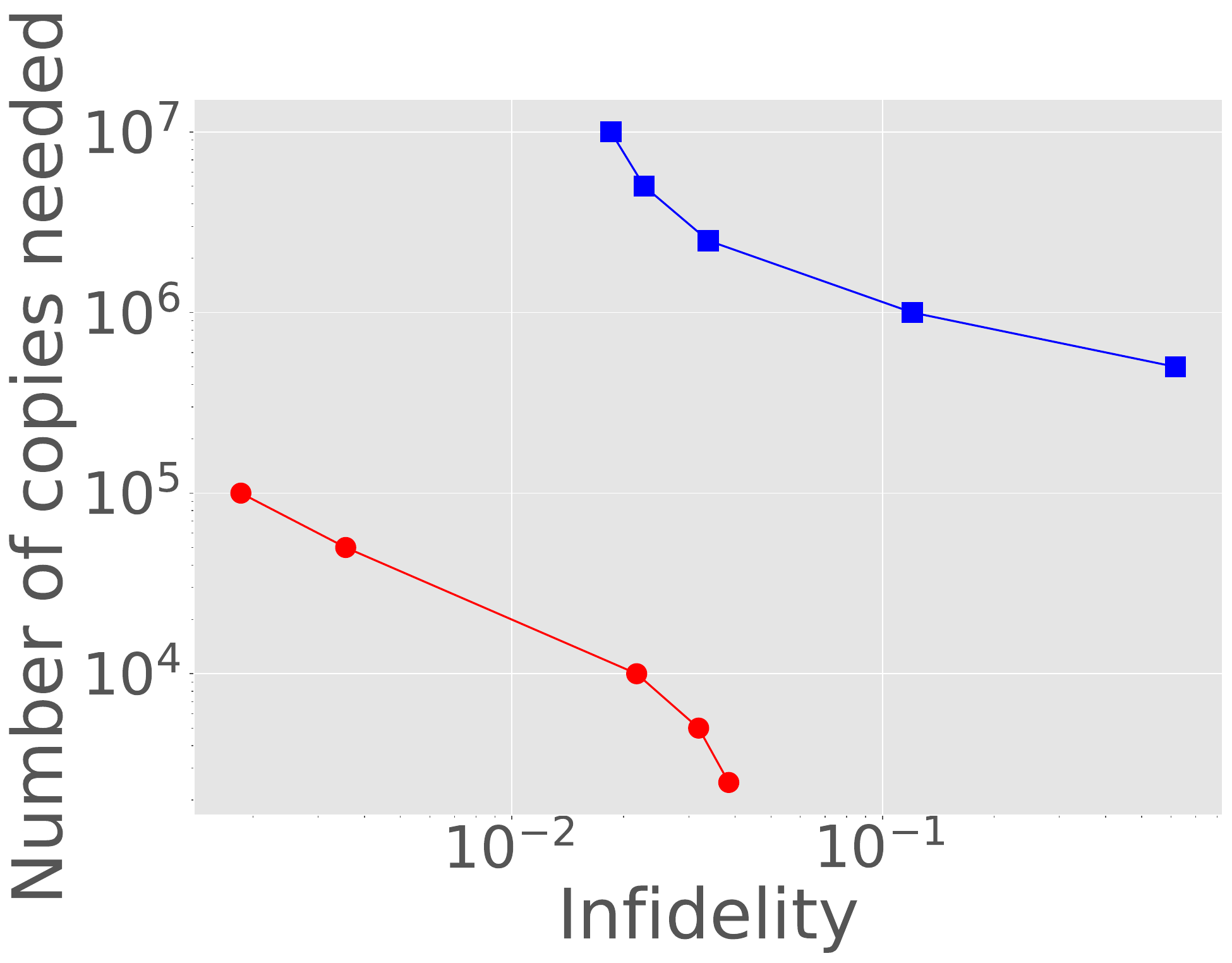} 
         &
        \includegraphics[width=0.5\columnwidth]{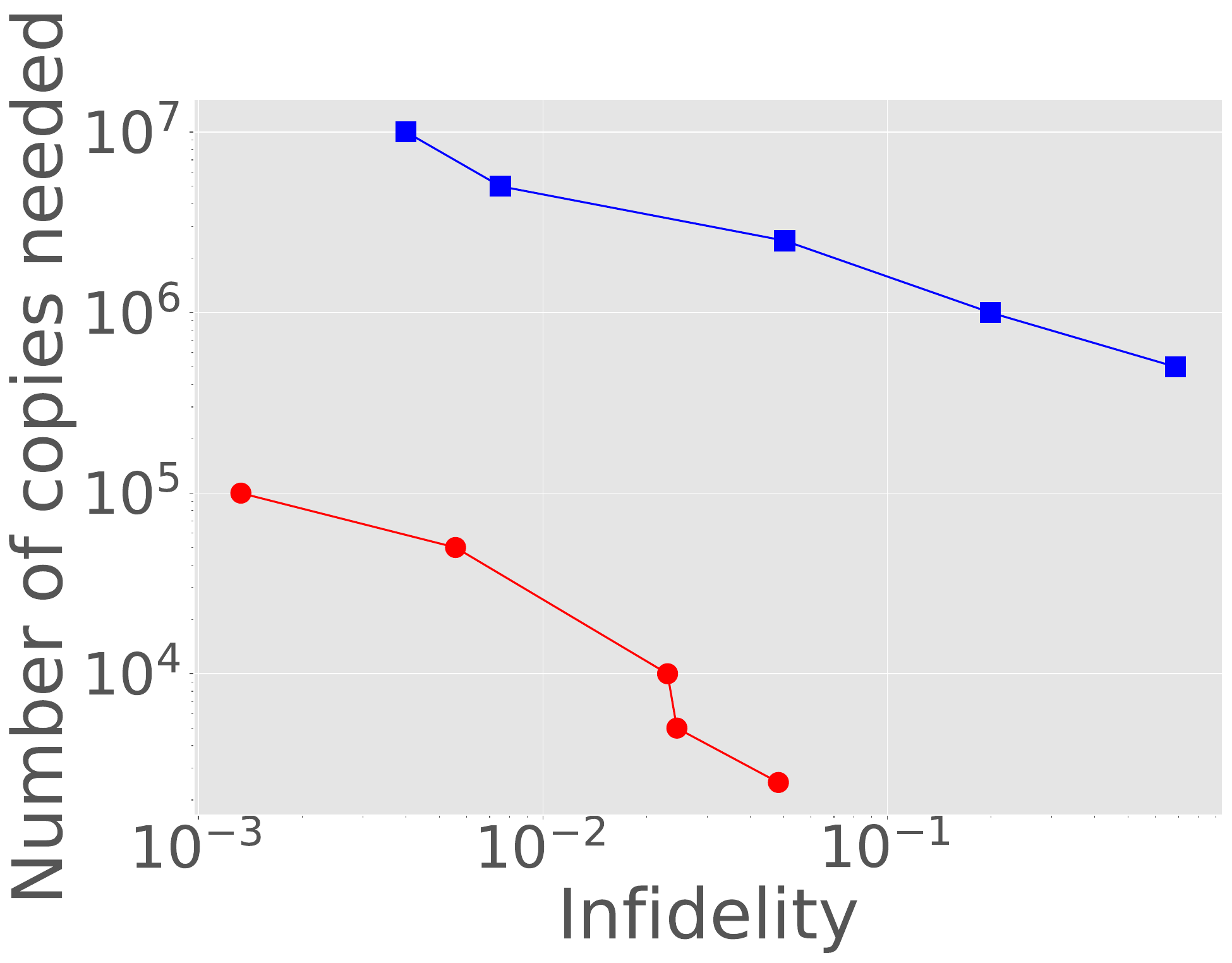}
        &
        \includegraphics[width=0.5\columnwidth]{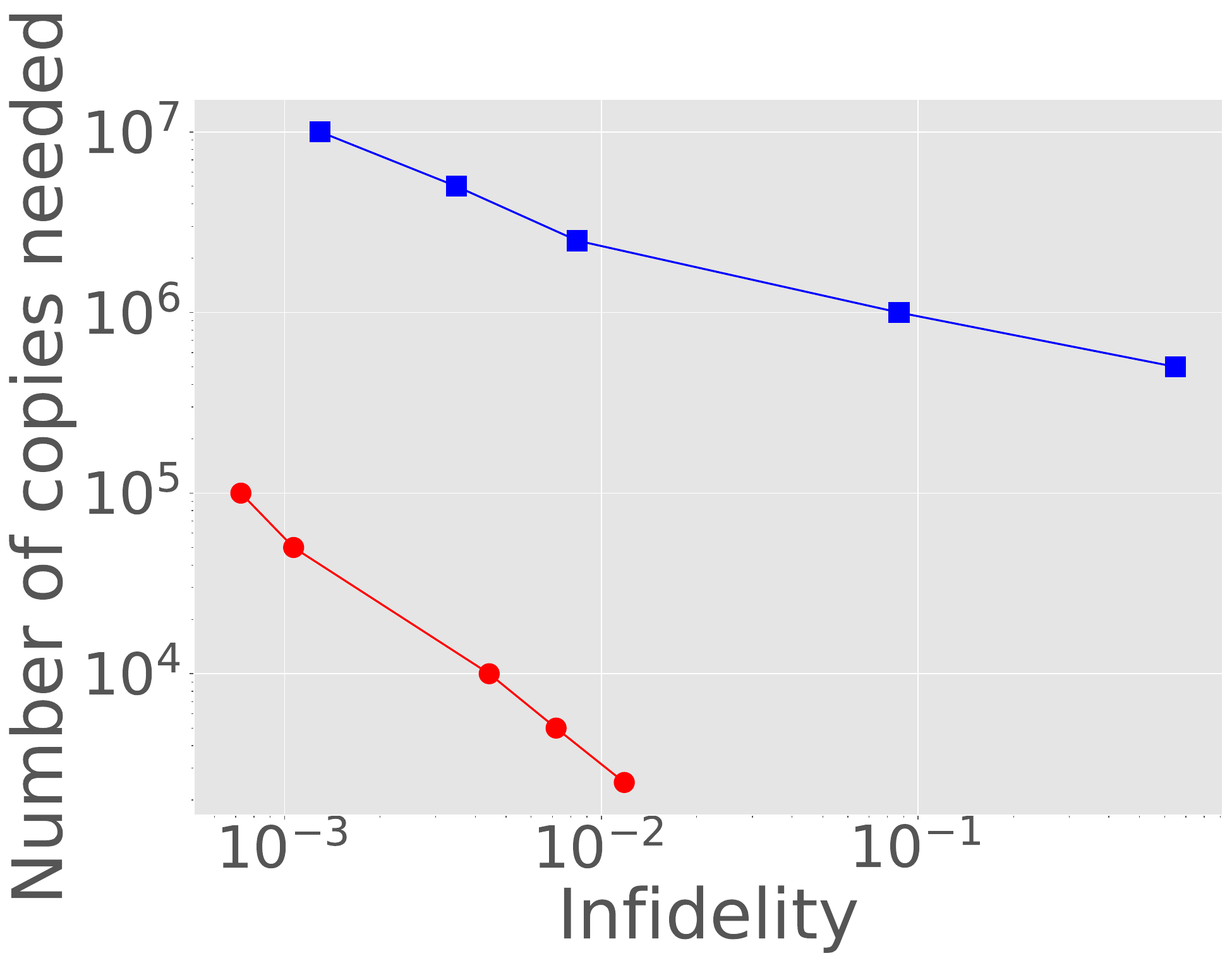}
        &
        \includegraphics[width=0.5\columnwidth]{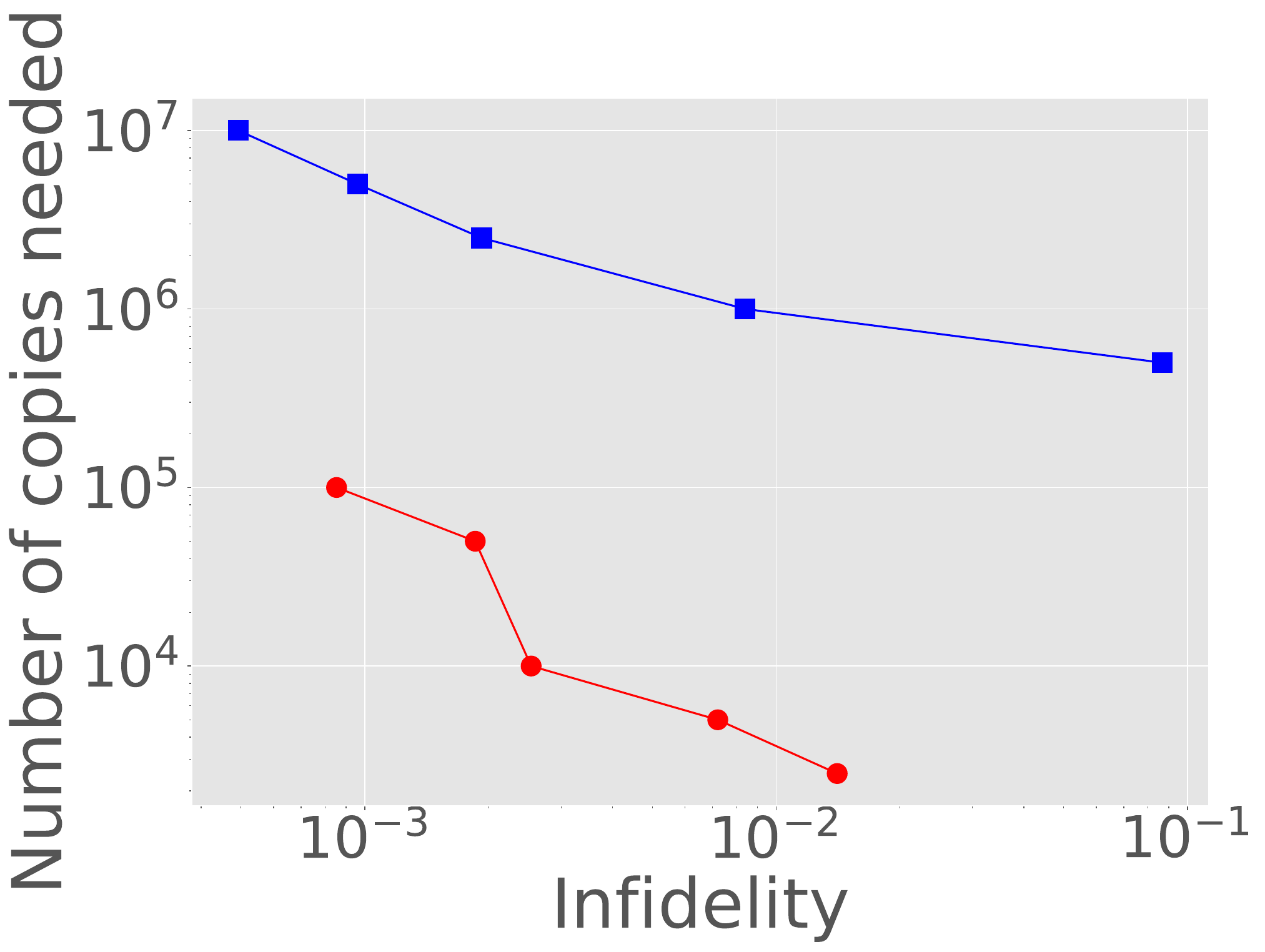} \\
        (a) & (b) & (c) & (d) \\
         \includegraphics[width=0.5\columnwidth]{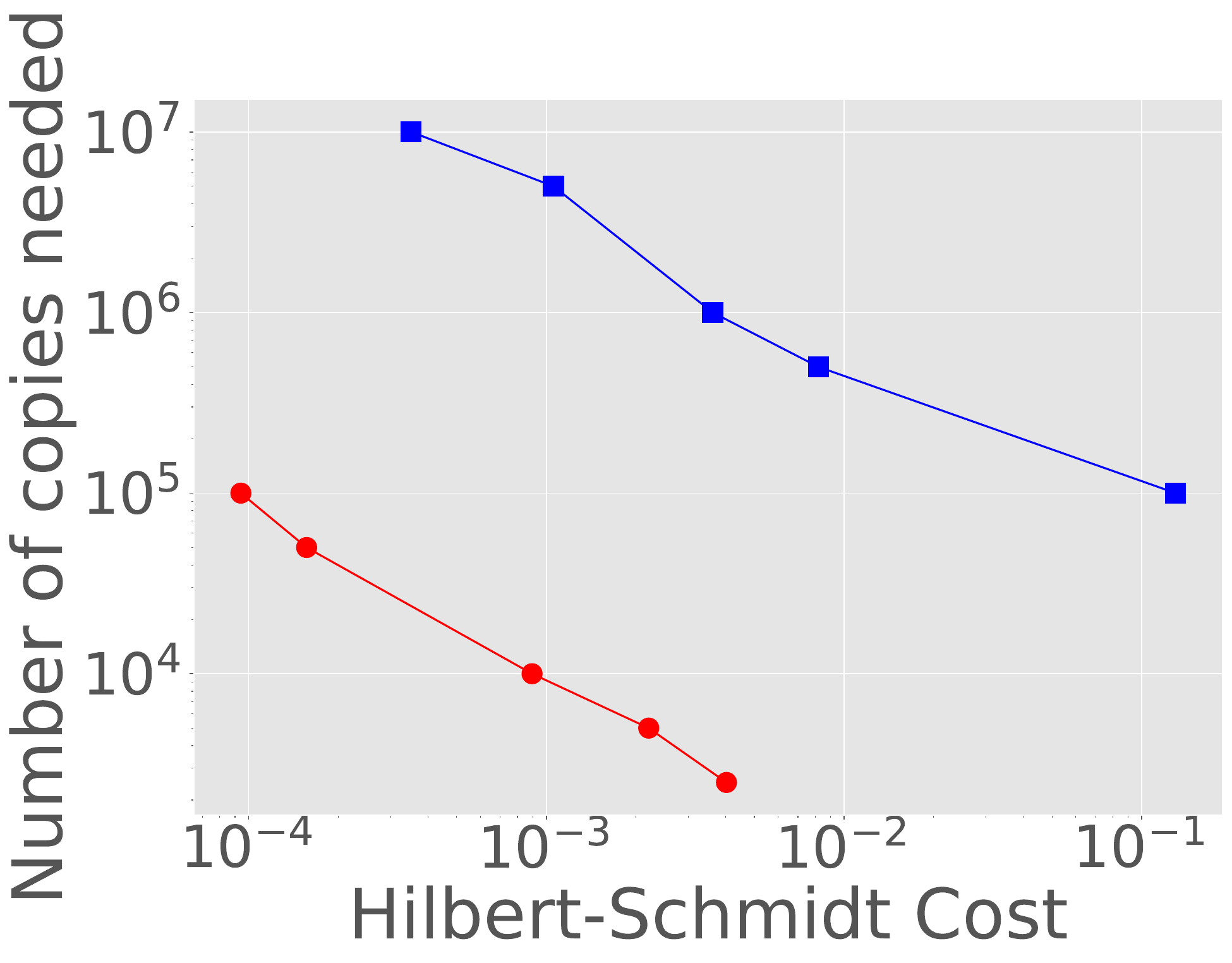} 
         &
        \includegraphics[width=0.5\columnwidth]{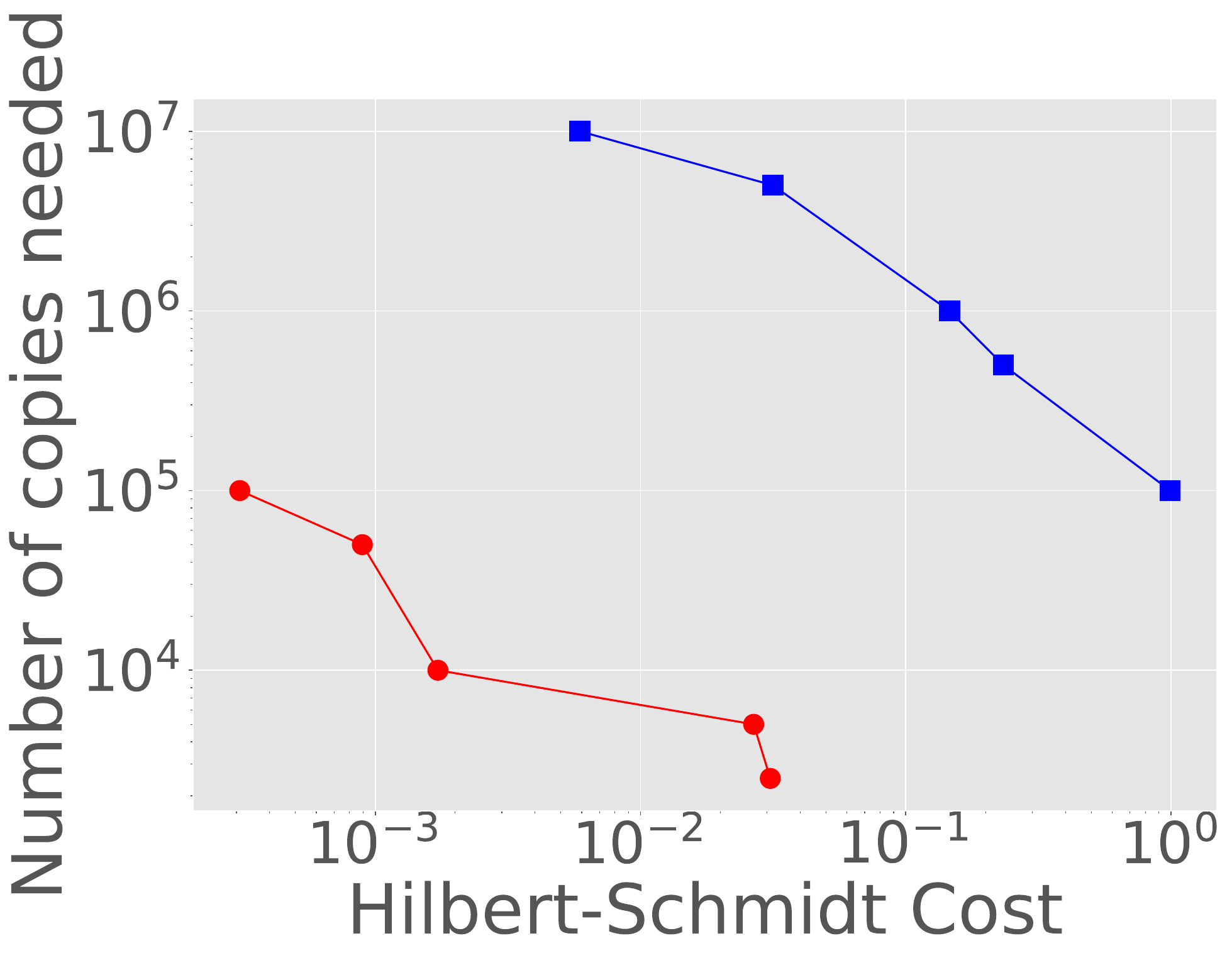}
        &
        \includegraphics[width=0.5\columnwidth]{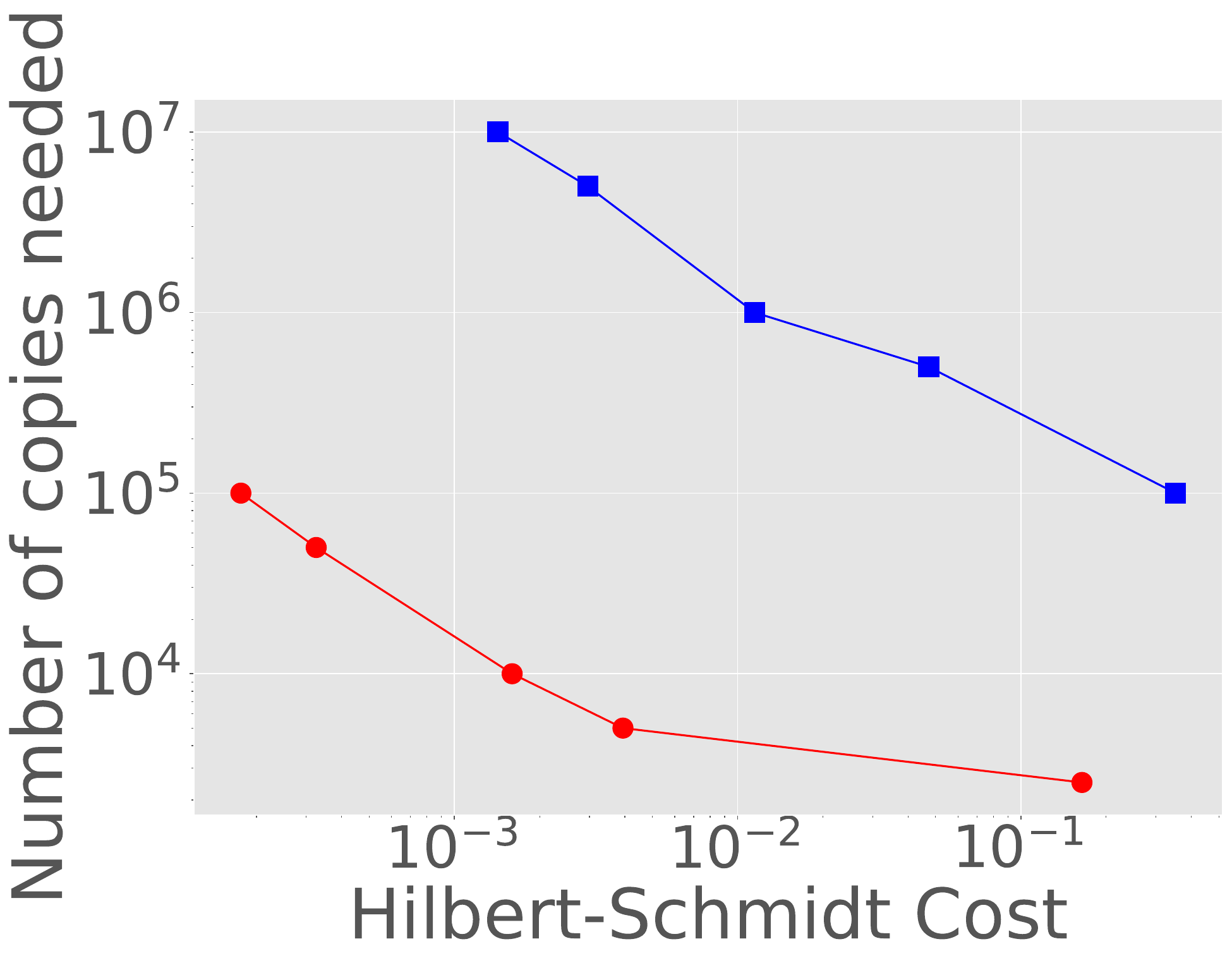}
        &
        \includegraphics[width=0.5\columnwidth]{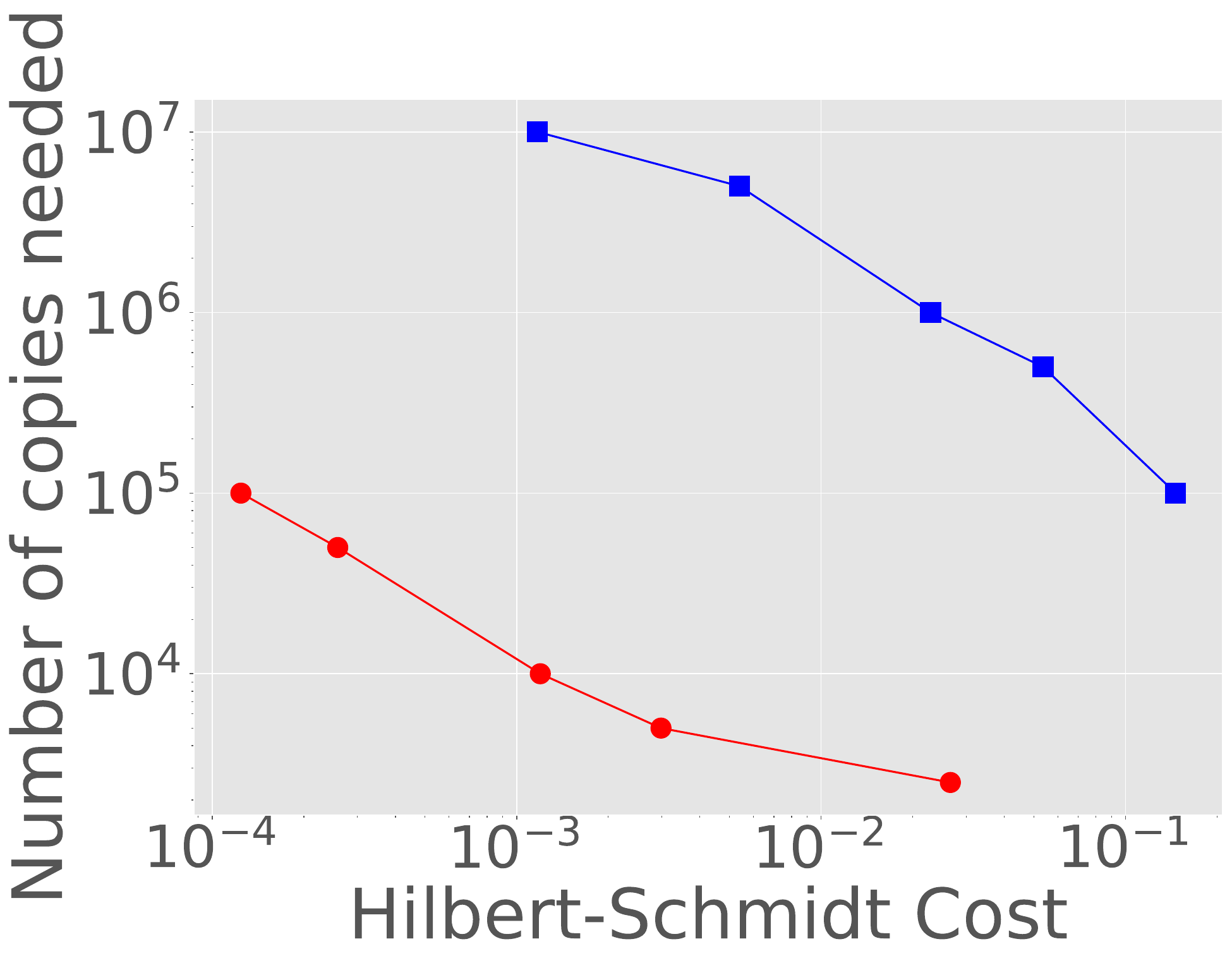} \\
        (e) & (f) & (g) & (h)
        \end{tabular}
        \caption{Resource needs for different infidelity objectives. All points plotted correspond to the mean of $5$ instances of the problem, with x-axis representing average lowest infidelity/Hilbert-Schmidt Cost achieved and y-axis representing the total number of copies consumed to achieve it. The classical optimizers used are the same as Figure~\ref{fig:fin_vs_AISO}. Plots (a,e), (b,f), (c,g), (d,h) correspond to ALA, MERA, HEA and TTN being used as the ansatz respectively. The order of magnitude savings in the number of copies when using AISO is evident. } \label{fig:xaxis}
    \end{figure*}

\section{Simulation Results} \label{sec:simulation_results}

Here we elaborate on the experimental results by comparing the sample complexity of AISO and the standard VQA in the two use cases discussed above. Python code to replicate our experiments can be found in~\cite{github}.

The depth $d$ of the shallow shadow ensemble (cf. Figure~\ref{fig:shadow_circuits}) is set to $3$ throughout the experiments. The viability of AISO in solving both problems is tested across four different ansatzes whose structures are given in Figure~\ref{fig:circuits}(a,b,c,d). Except in HEA, all two-qubit gates can be arbitrary two-qubit subcircuits. The specific ones used in our simulation are given in Figure~\ref{fig:circuits}(e). Also, for VQCS, each two-qubit subcircuit is a combination of two of these. In HEA, the two-qubit gate used is the CNOT gate.

For VQSP, we have used the Simultaneous Perturbation Stochastic Approximation~\cite{Spall1992} (SPSA), where the converging sequences used are, respectively, $c_r=a_r=r^{-0.4}$ and the total number of iterations is $5000$. On the other hand, the results of VQCS have used Powell's method~\cite{Powell1964} with a maximum of $10^3$ function evaluations allowed. We denote by AISO/VQA ($T$) the AISO/VQA algorithm that uses $T$ copies in total. This means that VQA $(T)$ will consume $T/10^4$ copies per function evaluation in SPSA and $ T/10^3$ copies in Powell's method. This is because SPSA requires two function evaluations to produce estimates of the gradient.

The unknown target states considered in the VQSP are $8$-qubit states, which are also compatible with the corresponding ansatzes being used. In each setting, the experiment is carried out across five different states and the results are shown in Figure~\ref{fig:fin_vs_AISO}(a-d). Here, we have plotted the mean of infidelity values achieved at different iterations across the five different experiments that were carried out. The shaded region comprises the mean plus and minus $0.3$ times the standard deviation of the five different infidelities.

In Figure~\ref{fig:fin_vs_AISO}(a-d), VQA $(5 \times 10^5)$, which utlizes $5 \times 10^5$ copies in total, consumes $50$ state copies per function evaluation. Similarly, the other VQA algorithms consume $100$ and $250$ state copies per evaluation. One can see that AISO closely matches or outperforms the results of VQA by consuming only $10^4$ copies in total.

Moving on to VQCS, similar experiments are carried out for $4$-qubit quantum gates (meaning $8$-qubits used in total). The results are summarized in Figure~\ref{fig:fin_vs_AISO}(e-h). Here, the minimum $H(\boldsymbol{\theta})$ in each interval of $10^2$ function evaluations out of the total allowed $10^3$ is plotted. The three VQA algorithms used here consume $ 10^2, 10^3$ and $ 10^4$ copies per function evaluation respectively. It is clear from the plots that AISO can match the performance of standard VQA similarly using considerably fewer copies to what we saw in the case of VQSP. 

In Figure~\ref{fig:xaxis}, we present the superiority of AISO over VQA in a different light. On the x-axis, we plot different infidelity or Hilbert-Schmidt cost values, and on the y-axis, we plot the number of copies required to achieve them, which are exponentially better for AISO.

    \begin{figure}[t] 
    \centering
    \begin{tabular}{cc}
         \includegraphics[width=0.35\columnwidth]{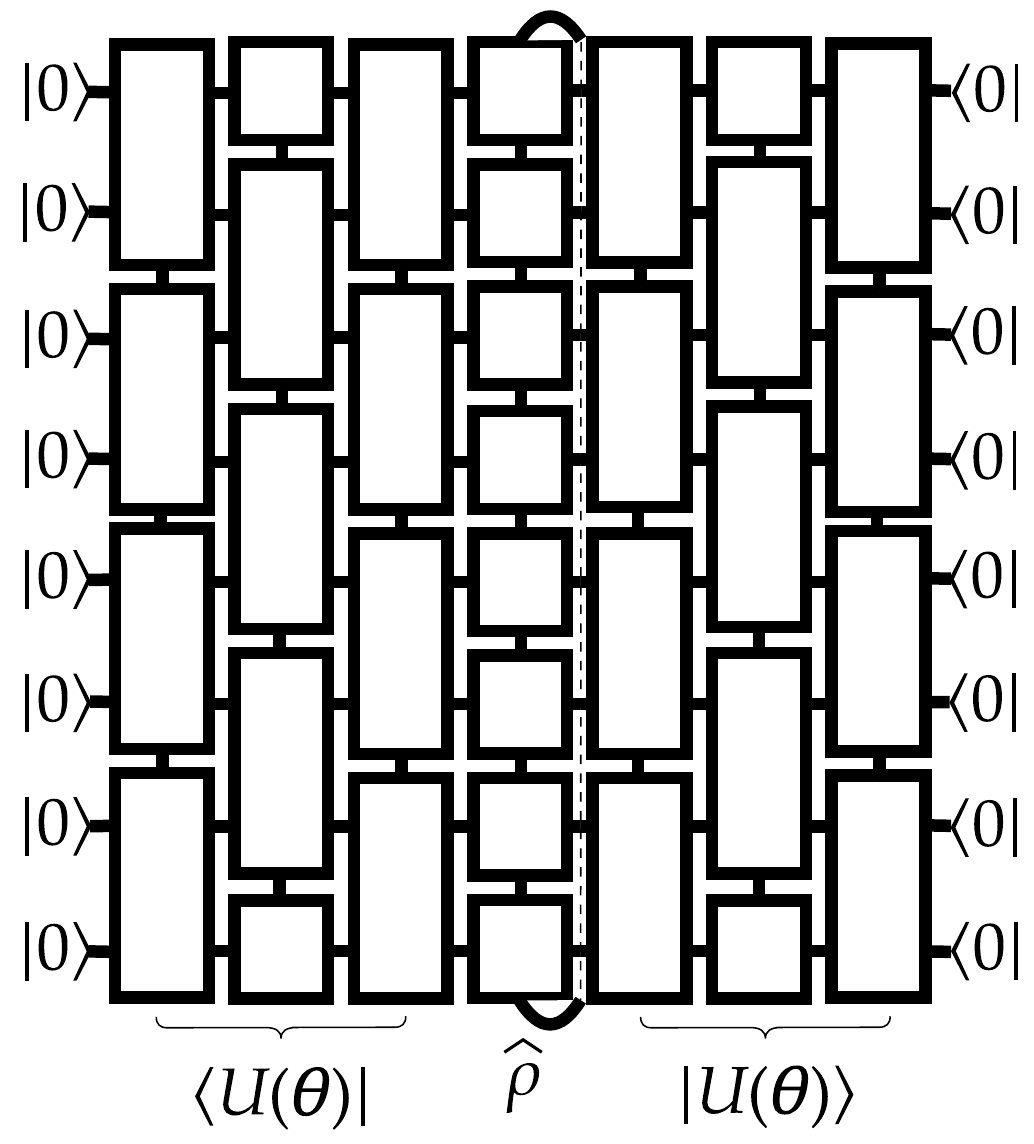} 
         &
        \includegraphics[width=0.33\columnwidth]{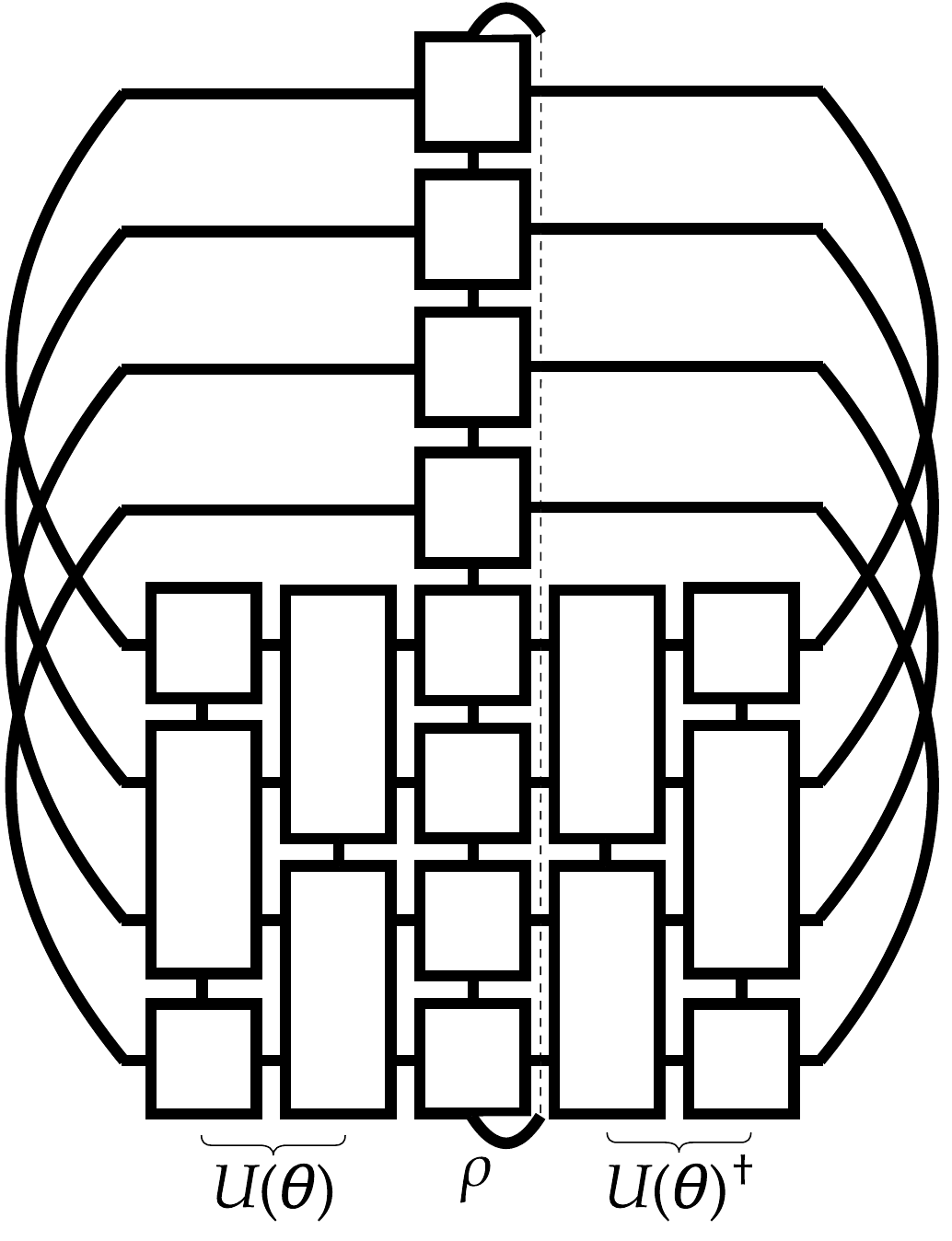} \\
        (a) State preparation & (b) Circuit Synthesis 
    \end{tabular}
    \caption{Tensor networks to compute $\langle W_O(\boldsymbol{\theta})\rangle_{\hat{\rho}}$. The examples used here uses the ALA. (a) corresponds to VQSP while (b) corresponds to VQCS. To contract (a) efficiently, we can start from the top qubit wire and contract wire by wire. One can see that, at every step, the total number of free indices the tensor will have is $\mathcal{O}(\log n)$, thus the cost of contraction is $\mathcal{O}(\text{poly}(n))$. Note that this is true for any ansatz with $R_U \in \mathcal{O}(\log n)$. A similar argument can be made for (b) when we start contracting ring by ring from the top.} \label{fig:tensor_networks}
    \end{figure} 
\section{Improved Bounds Using 2-Design Assumption} \label{sec:impact_2_design}

    In this section, we analyze the assumption of the input state in more detail. The assumption that the state is sampled from a $1$-design merely says that the input state is the maximally mixed state. So, to further understand the notion of a ``typical input state" and to get closer to the notion of the input state being an average state or a randomly generated state, we make a stronger assumption on the distribution. More precisely, we assume that the input state is sampled from a \textit{state $2$-design} $\mathcal{D}_2$. These are ensembles such that sampling from them is equivalent to sampling a pure state uniformly up to two statistical moments. $2$-designs are extensively used in quantum information to generate pseudorandomness and to analyze average case complexities~\cite{Dankert2009,Scott2008,Ambainis2009}.

    In this regime, we derive two results, starting with an upper bound on the variance of the state-dependent shadow norm when the state is sampled from a state $2$-design.

    \begin{theorem} \label{th:var}
        Let $\mathcal{D}_2$ be a state $2$-design and $d = \Theta(\log n)$. Then, for any observable $O$, we have

        \begin{align}
            \text{Var}_{\sigma \sim \mathcal{D}_2} \left( \| O \|_{\sigma,\mathcal{U}_d}^2 \right) \leq 64 \| O\|_F^2.
        \end{align}
    \end{theorem}

    Using this result, we can derive a result similar to Theorem~\ref{th:AISO}, with better constants.

        \begin{theorem} \label{th:AISO_2design}
            Let $d = \Theta(\log n)$ and $\rho$ be an $n$-qubit pure state sampled from a state 2-design $\mathcal{D}_2$. For any $\delta, \epsilon \in (0,1)$, $m > 1/\sqrt{\delta}$, and any $C>0$, let
            \begin{align}\label{eq:T_12design}
            T_1 \geq 2 \log \left(\frac{2(m^2-1)C}{m^2\delta-1}\right),\ T_2 \geq \frac{136}{\epsilon^2} (2m + 1)\| O\|_F^2.
            \end{align} 
            Then for any parameter vectors $\boldsymbol{\theta} ^ {(1)}, \boldsymbol{\theta} ^ {(2)}, \dots, \boldsymbol{\theta} ^ {(C)}$, all values $\langle W_O(\boldsymbol{\theta} ^ {(c)}) \rangle_{\rho}$, $1\leq c\leq C$, defined as in Eq~\eqref{eq:costfun} can be estimated using $ \langle \widehat{W_O}(\boldsymbol{\theta} ^ {(c)}) \rangle_{\rho}$ defined as in Eq~\eqref{eq:tfun} so that with probability at least $1-\delta$, we have $| \langle W_O(\boldsymbol{\theta} ^ {(c)})\rangle_{\rho} - \langle \widehat{W_O}(\boldsymbol{\theta} ^ {(c)}) \rangle_{\rho} | \leq \epsilon$
            for all $c$.
        \end{theorem}

        Hence, we see that the lower bound on $T_1$ in Eq~\eqref{eq:T_12design} is a constant time better than the lower bound on $T_1$ in Eq~\eqref{eq:T_11design}. By replacing the function evaluations in Theorem~\ref{th:AISO_2design} with expectations with arbitrary observables, one can see that similar advantages can be gained for regular shallow shadow estimation also when the input is sampled from a $2$-design.

\section{Dealing With Barren Plateaus} \label{sec:barren_plateaus}

    In some cases, the usage of global observables has been shown to introduce barren plateaus into the training landscape~\cite{Cerezo2021,Liu2022}. These are regions with gradients exponentially small in the number of qubits, which makes evaluating them using quantum devices extremely difficult. Several heuristic approaches have been proposed, which have been experimentally shown to be effective in certain cases. Even though AISO uses global observables, we note that our method is compatible with almost all barren plateau mitigating heuristic methods that have been proposed in the literature. For example,~\cite{Patti2021,Mele2022,Rad2022,Skolik2021,Grimsley2023,Grimsley2023_adapt,Friedrich2022,Verdon2019,Grant2019,Kulshrestha2022,Zhang2022} are methods that ultimately use the quantum device only to estimate $ \langle W_O(\boldsymbol{\theta}) \rangle_{\rho}$ at certain carefully chosen inputs $\boldsymbol{\theta}$. So, it is clear that if we use shadows to estimate them, then exponential advantages similar to the ones discussed in this paper can be achieved.
    
\section{Conclusion and Future Direction} \label{sec:conclusion}
    In this work, we proposed AISO --- a training algorithm that leverages shallow shadows to achieve an exponential reduction in quantum resources required to train VQA cost functions involving almost any shallow ansatz and observables with low Frobenius norm. This allows one to do more iterations of the classical optimizer, more hyperparameter tuning, and experiment with ansatzes and optimizers with very few executions of the quantum device. We demonstrate this advantage in two important use cases of interest in quantum information: Variational Quantum State Preparation and Variational Quantum Circuit Synthesis. 

    In terms of future directions, we are trying to design similar resource-efficient ansatz agnostic protocols for local observables, by leveraging classical machine learning with classical shadows similar to~\cite{Huang2021}. We are also investigating the potential of generative machine learning-based tomography models such as~\cite{Ma2023,Cha2022,Zhong2022,Schmale2022} to improve the sample complexity of VQAs.

\section{Acknowledgement}
We thank Afham for pointing out important references and related works. This work is partially supported by the Australian Research Council (Grant No: DP220102059). AB was partially supported by the Sydney Quantum Academy PhD scholarship. 
\bibliography{references}

\section{Technical Appendix} \label{sec:appendix}
Here, we present the proofs of Theorems~\ref{th:classical_cost}, ~\ref{th:AISO},~\ref{th:var} and ~\ref{th:AISO_2design}. 

        \begin{customthm}{3}
            In AISO, for any quantum ansatz $U$ with $R_U \in \mathcal{O}(\log n)$,  $\langle \widehat{W_O}(\boldsymbol{\theta}) \rangle_{\rho}$ can be classically evaluated with cost $\mathcal{O}(\text{poly}(n) \cdot \log C \cdot \| O\|_F^2) $ for VQSP and VQCS. 
        \end{customthm}
    \begin{proof}
    The proof is based on the proof of  classical simulation of quantum circuits in~\cite{Jozsa2006}. We recall the main result of that paper here for brevity.

    \begin{customthm}{A1}
        \cite{Jozsa2006} Every quantum circuit $U$ can be classically simulated using tensor networks using cost $\mathcal{O}(n\cdot\text{poly}(2 ^ {R_U}))$.
    \end{customthm}
    
    The reasoning is that when we start contracting the tensor network from the top qubit-wire, the maximum number of free indices the tensor can have at any point will be $\mathcal{O}(R_U)$. This means that the size of the tensor at every point will be $\mathcal{O}(2^{R_U})$. So, such contractions for all $n$ qubit-wires can be done using cost $\mathcal{O}(n\cdot\text{poly}(2 ^ {R_U}))$.

    For VQSP of an unkwown state $\rho$, the required classical computation is mainly computing $ \langle W_{\ket{0} \bra{0}}(\boldsymbol{\theta}) \rangle_{\hat{\rho}}$ for some shallow shadow $\hat{\rho}$. An example of this can be seen in Figure~\ref{fig:tensor_networks}(a). But one can see that this is almost the same as a quantum circuit tensor network. The only difference is the fact that the core tensors of $\hat{\rho}$ may not be valid quantum operations. But that does not affect the complexity of tensor contraction. Since the shadows are generated using an ensemble with $d=\mathcal{O}(\log n)$, the core tensors will have dimension $\mathcal{O}(\text{poly}(n))$. Combining this with the fact that there are $\mathcal{O}(\log C \cdot \| O\|_F^2)$ shadows, the complexity of classical computation is $ \mathcal{O}(\text{poly}(n) \cdot \log C \cdot \| O\|_F^2)$.
    
    \begin{figure}[t] 
    \centering
    \begin{tabular}{cc}
         \includegraphics[width=0.4\columnwidth]{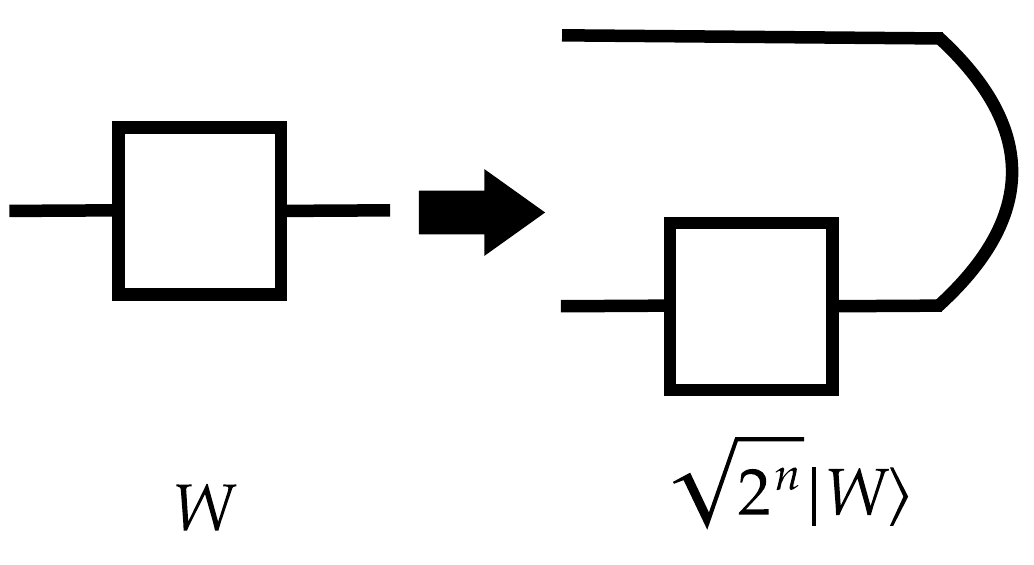} 
    \end{tabular}
    \caption{Vectorization of a unitary $W$.
    } \label{fig:vectorization}
    \end{figure}
    
    For VQCS of an unknown circuit $V$, the expectation that we are estimating is given as 

    \begin{align}
        \frac{\left| \text{tr}(U(\boldsymbol{\theta}) ^ {\dag} V) \right| ^ 2}{4^n} &= \left| \braket{U(\boldsymbol{\theta}) | V} \right|^2 \\
        &= \bra{U(\boldsymbol{\theta})} \ket{V} \bra{V} \ket{U(\boldsymbol{\theta})}.
    \end{align}
    
    So for any shadow $\hat{\rho}$ of the state $\ket{V}$, we have 
    \begin{align}
        \langle \ket{U(\boldsymbol{\theta})} \bra{U(\boldsymbol{\theta})} \rangle_{\hat{\rho}} = 
        \bra{U(\boldsymbol{\theta})} \hat{\rho} \ket{U(\boldsymbol{\theta})}.
    \end{align}

    Note that for any unitary matrix $W$, $\ket{W}$ is simply the vectorized and normalized version of $W$. If any unitary is depicted as a tensor block, we can always get a vectorized version by bending the output wire as shown in Figure~\ref{fig:vectorization}. So, given the circuit description of $U(\boldsymbol{\theta})$, we can get a tensor network depicting $\sqrt{2^n}\ket{U(\boldsymbol{\theta})}$ by bending the output wires in the manner shown in Figure~\ref{fig:tensor_networks}(b). Similar to how Figure~\ref{fig:tensor_networks}(a) can be contracted efficiently if we start from the top and go in a line-by-line manner, this network can also be contracted efficiently if we start from the top and contract ring by ring. That is, contract tensors along the top ring, then contract tensors along the second ring, and so on. Finally, we divide the answer by $2^n$.  Similar to the previous case, we see that at every instance, the number of free indices will be $\mathcal{O}(\log n)$, and hence the complexity of contracting this tensor is $\mathcal{O}(\text{poly}(n))$. Hence, since we have $\mathcal{O}(\log C \cdot \| O\|_F^2)$ shadows, the complexity of a single function evaluation is $ \mathcal{O}(\log C \cdot \| O\|_F^2 \cdot \text{poly}(n))$.

    It is easy to see that one can generalize this to arbitrary observables that can be represented as a quantum circuit-like tensor network with $R_O \in \mathcal{O}(\log n)$.

    \end{proof}
        \begin{customthm}{4} 
                        Let $d = \Theta(\log n)$ and $\rho$ be an $n$-qubit pure state sampled from a state 1-design $\mathcal{D}_1$. For any $\delta, \epsilon \in (0,1)$, $m > 1/\delta$, and any $C>0$, let
            \begin{align}
            T_1 \geq 2 \log \left(\frac{2(m-1)C}{m\delta-1}\right),\ T_2 \geq \frac{136}{\epsilon^2} m \| O\|_F^2.
            \end{align} 
            Then for any parameter vectors $\boldsymbol{\theta} ^ {(1)}, \boldsymbol{\theta} ^ {(2)}, \dots, \boldsymbol{\theta} ^ {(C)}$, all values $\langle W_O(\boldsymbol{\theta} ^ {(c)}) \rangle_{\rho}$, $1\leq c\leq C$, defined as in Eq~\eqref{eq:costfun} can be estimated using $ \langle \widehat{W_O}(\boldsymbol{\theta} ^ {(c)}) \rangle_{\rho}$ defined as in Eq~\eqref{eq:tfun} so that with probability at least $1-\delta$, we have $| \langle W_O(\boldsymbol{\theta} ^ {(c)})\rangle_{\rho} - \langle \widehat{W_O}(\boldsymbol{\theta} ^ {(c)}) \rangle_{\rho} | \leq \epsilon$
            for all $c$.

        \end{customthm}

    \begin{proof}
        As we established earlier, all $C$ function evaluations are expectations of $\rho$ with $C$ parameterized observables $W_O(\boldsymbol{\theta}^{(c)})$, each with $ \| W_O(\boldsymbol{\theta}^{(c)})\|_F = \| O\|_F$.
        
        From Theorem~2, we know that 
        \begin{align}
            \| \widetilde{O} \|_{\mathds1/2^n,\mathcal{U}_d}^2 &\leq 4 \| \widetilde{O} \|_F^2 \\
            \implies \mathbb{E}_{\rho \sim \mathcal{D}_1 }\|\widetilde{O} \|_{\rho, \mathcal{U}_d}^2 &\leq 4 \| \widetilde{O} \|_F^2
        \end{align}
        So by using Markov Inequality, we have
        \begin{align}
            &\text{Prob} \left[ \|\widetilde{O} \|_{\rho, \mathcal{U}_d}^2 \leq m \mathbb{E}_{\rho \sim \mathcal{D}_1 }\|\widetilde{O} \|_{\rho, \mathcal{U}_d}^2 \right] \geq 1 - 1/m \\
            \implies & \text{Prob} \left[ \|\widetilde{O} \|_{\rho, \mathcal{U}_d}^2 \leq 4m \|\widetilde{O} \|_F^2 \right] \geq 1 - 1/m
        \end{align}
        
        So with probability at least $1-1/m$, the state-dependent shadow norm is bounded by $4m\| \widetilde{O}\|_F^2$. So, for any $\delta', \epsilon \in (0,1)$, if we use $T_1 T_2$ shallow shadows, where $T_1=2\log(2C/\delta')$, $T_2 = (136m/\epsilon^2) \| \widetilde{O_i}\|_F^2$, with probability at least $ \left( 1 - \delta' \right)
        \left( 1 - 1/m \right)$, for all $c$, we will have $| \langle W_O(\boldsymbol{\theta})\rangle_{\rho} - \langle \widehat{W_O}(\boldsymbol{\theta})\rangle_{\rho}| \leq \epsilon$.

        Set $1 - \delta = (1 - \delta')(1-1/m)$. So we have

        \begin{align}
            &1 - \delta = (1 - \delta')(1-1/m) \\
            \implies &\delta' = 1 - \left(\frac{1-\delta}{1-1/m} \right) \\
            &\ \ \ = 1 + \frac{m(\delta - 1)}{m - 1} \\
            &\ \ \ = \frac{m\delta-1}{m-1}
        \end{align}
        This completes the proof.
    \end{proof}

    \begin{customthm}{5} 
        Let $\mathcal{D}_2$ be a state $2$-design and $d = \Theta(\log n)$. Then for any observable $O$, we have

        \begin{align}
            \text{Var}_{\sigma \sim \mathcal{D}_2} \left( \| O \|_{\sigma,\mathcal{U}_d}^2 \right) \leq 64 \| O\|_F^4.
        \end{align}
    \end{customthm}

    \begin{proof}
        Let $\mathcal{U}_d = \{ U_1, U_2, \dots, U_{| \mathcal{U}_d|} \} $. Similar to state $2$-designs, special ensembles called \textit{unitary $2$-designs} can be used to emulate sampling a unitary uniformly at random (according to the Haar measure) upto two statistical moments. Given any unitary $2$-design $\mathcal{W}$,  $\{ W\ket{0} \ | \ \forall \ W \in \mathcal{W}\}$ is a state $2$-design~\cite{Watrous2018}. Using this, we have 
        \begin{align*}
            &\text{Var}_{\rho \sim \mathcal{D}_2} \left( \| O\|_{\rho, \mathcal{U}_d}^2\right) \leq \mathbb{E}_{\rho \sim \mathcal{D}_2} \left( \| O\|_{\rho, \mathcal{U}_d} ^ 4 \right) \numberthis \\
            &= \mathbb{E}_{W \sim \mathcal{W}}  \left( \| O\|_{W\ket{0} \bra{0} W^{\dag}, \mathcal{U}_d} ^ 4 \right) \numberthis \\
            &= \mathbb{E}_{W \sim \mathcal{W}} \left[ \mathbb{E}_{U \sim \mathcal{U}_d} \sum \limits_{u=0 }^{2^n-1} \bra{u} U W\ket{0} \bra{0} W^{\dag} U^{\dag} \ket{u} \langle O\rangle_{\hat{\rho}_{U,u}} ^ 2 \right]^2 \numberthis \\
            &= \mathbb{E}_{W \sim \mathcal{W}} \mathbb{E}_{U_1 \sim \mathcal{U}_d} \mathbb{E}_{U_2 \sim \mathcal{U}_d} \sum \limits_{u_1, u_2=0}^{2^n-1} \bra{u_1} U_1 W \ket{0}
            \\
            &\bra{0} W^{\dag} U_1 ^ {\dag} \ket{u_1} \langle O\rangle_{\hat{\rho}_{U_1,u_1}} ^ 2
            \bra{u_2} U_2 W \ket{0} \bra{0} W^{\dag} U_2 ^ {\dag} \ket{u_2} \langle O\rangle_{\hat{\rho}_{U_2,u_2}} ^ 2 \numberthis \\
            &= \mathbb{E}_{U_1 \sim \mathcal{U}_d} \mathbb{E}_{U_2 \sim \mathcal{U}_d} \sum \limits_{u_1, u_2=0}^{2^n-1} \langle O\rangle_{\hat{\rho}_{U_1,u_1}} ^ 2 \langle O\rangle_{\hat{\rho}_{U_2,u_2}} ^ 2 \mathbb{E}_{W \sim \mathcal{W}} \\
            &\hspace*{1.5cm}\text{tr} \left[ W\ket{0} \bra{0} W^{\dag} U_1^{\dag} \ket{u_1} \bra{u_1} U_1 \right] \\
            &\hspace*{1.5cm}\text{tr} \left[ W\ket{0} \bra{0} W^{\dag} U_2^{\dag} \ket{u_2} \bra{u_2} U_2 \right] \numberthis
        \end{align*}

        Now, we recall Lemma 3 from~\cite{Cerezo2021}, which will be useful when taking expectations with respect to unitary $2$-deisgns.

        \begin{lemma}
            Let $\mathcal{W}$ be a unitary $2$-design of operators acting on $ \mathbb{C} ^ {2 ^ n}$ and let $A,B,C,D \in \mathcal{L}(\mathbb{C} ^ {2 ^ n})$ be arbitrary matrices. Then we have

            \begin{align*}
                &\mathbb{E}_{W \sim \mathcal{W}} \text{tr} \left[ WA W^{\dag} B \right] \text{tr} \left[ WC W^{\dag} D \right] \\
                &= \frac{1}{4 ^ n - 1} \left[ \text{tr} (A) \text{tr} (B) \text{tr} (C) \text{tr} (D) + \text{tr}(AC) \text{tr}(BD)\right] \\
                &-\frac{1}{2^n(4^n - 1)} \left[ \text{tr} (AC) \text{tr} (B) \text{tr} (D) + \text{tr} (A) \text{tr} (C) \text{tr} (BD)\right] \numberthis
            \end{align*}
        \end{lemma}
        We shall derive a simple corollary which will be useful in the proof.
        \begin{corollary}
            Let $\mathcal{W}$ be a unitary $2$-design of operators acting on $ \mathbb{C} ^ {2 ^ n}$ and let $A,B,C,D \in \mathcal{L}(\mathbb{C} ^ {2 ^ n})$ be pure states. Then we have 

            \begin{align}
                \mathbb{E}_{W \sim \mathcal{W}} \text{tr} \left[ WA W^{\dag} B \right] \text{tr} \left[ WC W^{\dag} D \right] \leq \frac{4}{4 ^ n} 
            \end{align}
        \end{corollary}

        \begin{proof}
            The third and fourth terms are non-negative and can be dropped since it involves only traces of states and fidelity between states. The first two terms are upper bounded by $1$ due to the same reason as well. Finally, consider the fact that $1/(4^n-1) \leq 2/4^n$  for any $n \geq 1$.
        \end{proof}
        
        Plugging this in (17) gives us

        \begin{align*}
            \text{Var}_{\rho \sim \mathcal{D}_2} \left( \| O\|_{\rho, \mathcal{U}_d}^2\right) \leq &\frac{4}{4^n}\mathbb{E}_{U_1 \sim \mathcal{U}_d} \mathbb{E}_{U_2 \sim \mathcal{U}_d} \\
            &\sum \limits_{u_1, u_2=0}^{2^n-1} \langle O\rangle_{\hat{\rho}_{U_1, u_1}} ^ 2 \langle O\rangle_{\hat{\rho}_{U_2, u_2}} ^ 2 \numberthis
        \end{align*}
        Since a state $2$-design is also a state $1$-design, we have 
        \begin{align}
            &\mathbb{E}_{\rho \sim \mathcal{D}_2} \mathbb{E}_{U \sim \mathcal{U}_d} \sum \limits_{u=0}^{2^n-1} \bra{u} U \rho U^{\dag} \ket{u} \langle O\rangle_{\hat{\rho}_{U,u}} ^ 2 \\
            &= \frac{1}{2^n}\mathbb{E}_{U \sim \mathcal{U}_d} \sum \limits_{u=0}^{2^n-1}  \langle O\rangle_{\hat{\rho}_{U,u}} ^ 2 \\
            &= \| O\|_{\mathds1/2^n, \mathcal{U}_d}^2 \leq 4 \| O\|_F^2.
        \end{align}

        This completes the proof.
    \end{proof}

        \begin{customthm}{6}
            Let $d = \Theta(\log n)$ and $\rho$ be an $n$-qubit pure state sampled from a state 2-design $\mathcal{D}_2$. For any $\delta, \epsilon \in (0,1)$, $m > 1/\sqrt{\delta}$, and any $C>0$, let
            \begin{align}
            T_1 \geq 2 \log \left(\frac{2(m^2-1)C}{m^2\delta-1}\right),\ T_2 \geq \frac{136}{\epsilon^2} (2m + 1)\| O\|_F^2.
            \end{align} 
            Then for any parameter vectors $\boldsymbol{\theta} ^ {(1)}, \boldsymbol{\theta} ^ {(2)}, \dots, \boldsymbol{\theta} ^ {(C)}$, all values $\langle W_O(\boldsymbol{\theta} ^ {(c)}) \rangle_{\rho}$, $1\leq c\leq C$, defined as in Eq~\eqref{eq:costfun} can be estimated using $ \langle \widehat{W_O}(\boldsymbol{\theta} ^ {(c)}) \rangle_{\rho}$ defined as in Eq~\eqref{eq:tfun} so that with probability at least $1-\delta$, we have $| \langle W_O(\boldsymbol{\theta} ^ {(c)})\rangle_{\rho} - \langle \widehat{W_O}(\boldsymbol{\theta} ^ {(c)}) \rangle_{\rho} | \leq \epsilon$
            for all $c$. 
        \end{customthm}

    \begin{proof}
        Using Chebychev's Inequality, when we sample a state $\rho$ from a $ 2$ design, we have

        \begin{align*}
            &\text{Prob} \Big[ \left| \| \widetilde{O}\|_{\rho,\mathcal{U}_d}^2 - \| \widetilde{O}\|_{\mathds1/2^n,\mathcal{U}_d} ^ 2 \right| \\ 
            &\hspace*{1cm} \leq m\sqrt{\text{Var}_{\rho \sim \mathcal{D}_2} \| \widetilde{O}\|_{\rho,\mathcal{U}_d}^2 }\Big] \geq 1 - \frac{1}{m^2} \numberthis \\
            \implies &\text{Prob} \left[  \| \widetilde{O}\|_{\rho,\mathcal{U}_d}^2 - 4\| O\|_F ^ 2  \leq 8 \| \widetilde{O}\|_F^2 m\right] \geq 1 - \frac{1}{m^2} \numberthis \\
            \implies &\text{Prob} \left[  \| \widetilde{O}\|_{\rho,\mathcal{U}_d}^2 \leq 4(2m + 1)\| \widetilde{O}\|_F ^ 2 \right] \geq 1 - \frac{1}{m^2} \numberthis
        \end{align*}
So with probability at least $1-1/m^2$, the state dependent shadow norm is bounded by $4(2m + 1)\| \widetilde{O}\|_F^2$. So, for any $\delta', \epsilon \in (0,1)$, if we use $T_1 T_2$ shallow shadows, where $T_1=2\log(2C/\delta')$, $T_2 = (136(2m + 1)/\epsilon^2) \| \widetilde{O_i}\|_F^2$, with probability $(1-\delta)(1-1/m^2)$, for all $c$, we will have $ | \langle W_O(\boldsymbol{\theta} ^ {(c)}) \rangle_{\rho} - \langle \widehat{W_O}(\boldsymbol{\theta} ^ {(c)}) \rangle_{\rho}| \leq \epsilon$.

        Set $1 - \delta = (1 - \delta')(1-1/m^2)$. So we have

        \begin{align}
            &1 - \delta = (1 - \delta')(1-1/m^2) \\
            \implies &\delta' = 1 - \left(\frac{1-\delta}{1-1/m^2} \right) \\
            &\ \ \ = 1 + \frac{m^2(\delta - 1)}{m^2 - 1} \\
            &\ \ \ = \frac{m^2\delta - 1}{m^2-1}
        \end{align}
        This completes the proof.
    \end{proof}

\end{document}